\documentclass[runningheads]{llncs}
\usepackage[utf8]{inputenc}
\usepackage{wrapfig}
\usepackage{amssymb,amsmath}
\usepackage[scr]{rsfso}
\usepackage{graphicx}
\usepackage{xcolor}
\usepackage{ifthen}
\usepackage{tikz}  
\usetikzlibrary{positioning,calc}
\usetikzlibrary{decorations.pathmorphing,shapes,decorations.pathreplacing}
\usepackage{enumitem}
\setlist{itemsep=0pt,topsep=4pt}

\newcommand{\comment}[4][inline]{
  \ifthenelse{\equal{#1}{margin}}
  {
    \marginpar{\scriptsize \hrule\noindent
    {\bf #2:\\} {\em\textcolor{#3}{#4}} \hrule}
  }
  {
    \hrule\noindent {\bf #2:} {\em\textcolor{#3}{#4}} \hrule
  }
}

\newcommand{\longversion}[1]{#1}

\renewcommand{\phi}{\varphi}
\newcommand{\powerset}[1]{\mathscr{P}(#1)}

\newcommand{\M}{\mathcal{M}}
\newcommand{\A}{\mathcal{A}}
\newcommand{\I}{\mathcal{I}}
\newcommand{\T}{\mathcal{T}}
\renewcommand{\O}{\mathcal{O}}
\newcommand{\IS}{\mathcal{MP}_N}
\newcommand{\B}{B} 
\newcommand{\W}{W} 

\newcommand{\vertices}[1]{\mathcal{V}(#1)}
\newcommand{\facets}[1]{\mathcal{F}(#1)}
\newcommand{\emptyview}{\Box}

\newcommand{\Agents}{\mathrm{Ag}}
\newcommand{\Atom}{\mathrm{At}}
\newcommand{\inputprop}[2]{\mathsf{input}_{#1}^{#2}}
\newcommand{\decideprop}[2]{\mathsf{decide}_{#1}^{#2}}
\newcommand{\pre}{\mathsf{pre}}
\newcommand{\extended}{\widehat}
\newcommand{\ExtAtom}{\extended{\Atom}}
\newcommand{\ExtProd}[2]{\extended{#1[#2]}}
\newcommand{\ExtO}{\extended{\O}}

\newcommand{\inputs}{\mathit{I}}

\def\cA{{\cal A}}

\def\cF{{\cal F}}
\def\cG{{\cal G}}
\def\cH{{\cal H}}
\def\cI{{\cal I}}

\def\cK{{\cal K}}
\def\cO{{\cal O}}
\def\cP{{\cal P}}

\def\cT{{\cal T}}
\def\cV{{\cal V}}
\def\cM{{\cal M}}
\def\cI{{\cal I}}

\def\Gz{\cI}

\newcommand{\set}[1]{\left\{ #1 \right\}}

\newcommand{\view}{\mathsf{view}}

\newcommand{\var}[1]{\lstinline+#1+}

\newcommand{\vs}{s}

\newcommand{\Vin}{\ensuremath{V^\textit{in}}}
\newcommand{\Vout}{\ensuremath{V^\textit{out}}}

\def\cA{{\cal A}}
\def\cV{{\cal V}}
\def\cF{{\cal F}}
\newcommand{\la}{\langle}
\newcommand{\ra}{\rangle}
\newcommand{\Values}{\mathit{Val}}
\newcommand{\Lcal}{\mathcal{L}}
\def\cK{{\cal K}}

\begin{document}
\title{A dynamic epistemic logic analysis of the equality negation task}
%
%
\author{\' Eric Goubault\inst{1} \and Marijana Lazi\' c\inst{2} \and J\' er\' emy Ledent\inst{1} \and Sergio Rajsbaum\inst{3}}
\authorrunning{E. Goubault et al.}
%
\institute{LIX, CNRS, \'Ecole Polytechnique, IP-Paris \and
TU M\" unchen \and
Instituto de Matem\'aticas, UNAM
}
\maketitle              
\begin{abstract}
In this paper we study the solvability of the \emph{equality negation} task in a simple wait-free
model where processes communicate by reading and writing shared variables or exchanging messages. 
 In this task, two processes start with a private input value in the set  $\set{0,1,2}$,
 and after communicating, each one must  decide a binary output value,  so that the outputs of the processes are the same if and only if the input values of the processes are different.
This task is already known to be unsolvable; our goal here is to prove this result using the dynamic epistemic logic (DEL) approach introduced by Goubault, Ledent and Rajsbaum in GandALF 2018.
We show that in fact, there is no epistemic logic formula that explains why the task is unsolvable. We fix this issue by extending the language of our DEL framework, which allows us to construct such a formula, and discuss its utility.
\end{abstract}

\section{Introduction}
\subsubsection{Background.}
Computable functions are the basic objects of study in computability theory. A function is computable if there exists a Turing machine which, given an input of the function domain, returns the corresponding output. If instead of one Turing machine, we have many, and each one gets only one part of the input, and should compute one part of the output, we are in the setting of \emph{distributed computability}, e.g.~\cite{AW04,Lynch96}. 
\longversion{
  The sequential machines are called \emph{processes}, and are  allowed to be
  infinite state machines,
  to concentrate on the interaction aspects of computability, disregarding
  sequential computability issues.
}
The notion corresponding to a function is a \emph{task}, roughly, the domain is a set of input vectors, the range is a set of output vectors, and the task specification $\Delta$ is an input/output relation between them.
An input vector~$I$ specifies in its $i$-th entry the (private) input to the $i$-th process,
and an output vector $O\in\Delta(I)$ states that it is valid for each process $i$ to produce as output the $i$-th entry of~$O$, whenever the input vector is~$I$. 
An important example of a task is \emph{consensus}, where each process is given an input from a set of possible input values, and the participating processes have to agree on one of their inputs.

A \emph{distributed computing model} has to specify various details related to how the processes communicate
with each other and what type of failures may occur.
It turns out that different models may have different power, i.e., solve different sets of tasks.
In this paper we consider the \emph{layered message-passing model}~\cite{HerlihyKR:2013}, both because of its
relevance to real systems, and because it is the basis to study task computability. This simple, wait-free round-based
model where messages can be lost, is described in Section~\ref{sec:prelim}.
 
The theory of distributed computability has been well-developed since the early 1990's~\cite{HerlihyS99}, with origins even before~\cite{BiranMZ90,FischerLP85},
 and  overviewed in a book~\cite{HerlihyKR:2013}.
It was discovered that the reason for why a task may or may not be computable is of a topological nature.
 The input and output sets of vectors are best described as \emph{simplicial complexes},
  and a task can be specified by a relation $\Delta$ from the input complex $\mathcal I$ to the output complex $\mathcal O$.
 The main result is that a task is solvable in the layered message-passing model  if and only if 
 there is a certain subdivision of the input complex ${\mathcal I}$ and a certain simplicial map $\delta$ to the output complex ${\mathcal O}$, that respects the specification $\Delta$. 
\longversion{ This is why the layered message-passing model is fundamental; models  that can solve more tasks than
 the layered message-passing model preserve the topology of the input complex less precisely (they introduce
 ``holes'').
 }

\subsubsection{Motivation.}
We are interested in understanding distributed computability from the epistemic point of view.
What is the knowledge that the processes should  gain, to be able to solve a task?
This question began to be addressed in~\cite{gandalf2018}, using dynamic epistemic logic (DEL).
Here is a brief overview of the approach taken  in~\cite{gandalf2018}.
A new \emph{simplicial complex model} for a multi-agent system was introduced, 
instead of the usual Kripke epistemic $S5$ model based on graphs.
Then, the initial knowledge of the processes is represented by a simplicial model,
denoted as~$\cI$, based on 
the input complex of the task to be solved. The distributed computing model is represented by an action model $\cA$,
and the knowledge at the end of the executions of a protocol is represented by 
the product update~$\cI [\cA ]$, another simplicial model.
Remarkably, the task specification is also represented by an action model $\cT$, and the
product update gives a simplicial complex model~$\cI [\cT ]$ representing the knowledge that should be acquired,
by a protocol solving the task.
 The task ${\mathcal{T}}$ is \emph{solvable} in $\mathcal{A}$ 
 whenever there exists a morphism $\delta: \Gz[\cA]\rightarrow \Gz[{\mathcal{T}}]$ such that
the diagram of simplicial complexes below commutes.
\begin{wrapfigure}[8]{r}{0.2\textwidth}
 \centering
 \vspace{-22pt}
\begin{tikzpicture}
  \node (s) {$\Gz[\cA]$};
  \node (xy) [below=2 of s] {$\Gz[\mathcal{T}]$};
  \node (x) [left=of xy] {$\Gz$};
  \draw[<-] (x) to node [sloped, above] {$\pi$} (s);
  \draw[->, dashed, right] (s) to node {$\delta$} (xy);
  \draw[->] (xy) to node [below] {$\pi$} (x);
\end{tikzpicture}\end{wrapfigure}

Thus, to prove that a task is unsolvable, one needs to show that no such $\delta$ exists. 
But one would want to produce a specific formula, that concretely represents knowledge that exists
in $\cI [\cT ]$, but has not been acquired after running the protocol, namely  in $\cI [\cA ]$.
Indeed, it was shown  in~\cite{gandalf2018} that  two of the main impossibilities in distributed computability,
consensus~\cite{FischerLP85,LA87} and approximate agreement~\cite{HerlihyKR:2013}, can be expressed by 
such a formula.
However, for other unsolvable tasks (e.g. set agreement), no such formula has been found, despite the fact that
no morphism~$\delta$ exists. 

\subsubsection{Contributions.}
In this paper we show that actually, there are unsolvable tasks, for which no such formula
exists, namely,  the \emph{equality negation} task, 
defined by Lo and Hadzilacos~\cite{DBLP:journals/siamcomp/LoH00} and studied
by  the authors in \cite{DISC19}.
This task was introduced as the central idea to prove that the consensus hierarchy~\cite{Herlihy:1991:WS:114005.102808,Jayanti:1993:RHH:164051.164070} is not robust.
\begin{center}
 \includegraphics[width=0.9\textwidth]{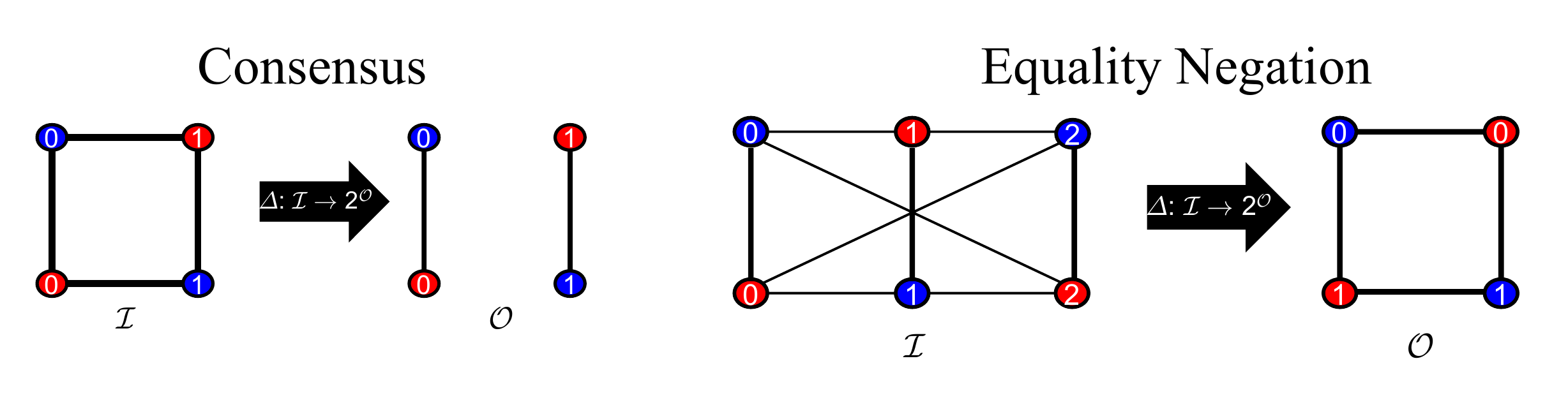}
 \vspace{-10pt}
\end{center}
Consider two processes $P_0$ and $P_1$, each of which has a private input value, drawn from the set 
of possible input values $\inputs=\{0,1,2\}$. 
After communicating, each process must irrevocably decide a binary output value, either $0$ or $1$, so that the outputs of the processes are the same if and only if the input values of the processes are different.

It is interesting to study the solvability of the equality negation task from the epistemic point of view.
It is well known that there is no wait-free consensus algorithm in our model~\cite{Chor:1987:PCU:41840.41848,LA87}.
The same is true for equality negation, as shown in~\cite{DBLP:journals/siamcomp/LoH00,DISC19}.
This is intriguing because there is a formula that shows the impossibility
of consensus (essentially reaching common knowledge on input values)~\cite{gandalf2018}, 
while, as we show here, there is no such formula for equality negation.
In more detail, 
it is well known that consensus is intimately related to connectivity, and hence to common knowledge, while its specification requires
deciding unto  disconnected components of the output complex.
The equality negation task is unsolvable for a different reason, since its output complex
is connected.
Moreover, equality negation is strictly weaker than consensus: consensus can implement equality negation,
but not viceversa (the latter is actually a difficult proof in~\cite{DBLP:journals/siamcomp/LoH00}).
So it is interesting to understand the difference between the knowledge required to solve each of these tasks.


%

Our second contribution is to
propose an extended version of our DEL framework, for which there is such a formula.
Intuitively, the reason why we cannot find a formula witnessing the unsolvability of the task is because our logical language is too weak to express the knowledge required to solve the task.
So, our solution is to enrich the language by adding new atomic propositions, allowing us to express the required formula.

\subsubsection{Organization.}
Section~\ref{sec:prelim} recalls the DEL framework introduced in~\cite{gandalf2018}, and defines the layered message-passing model in this context.
In Section~\ref{sec:negationTask} we study the equality negation task
using DEL. First we explain why the impossibility proof does not work in
the standard setting, then we propose an extension allowing us to make the proof go through.
The Appendix includes proofs and a detailed treatment of the equality negation task
following the combinatorial topology approach, for completeness,
but also for comparison with the DEL approach.

\section{Preliminaries}
\label{sec:prelim}

\subsection{Topological models for Dynamic Epistemic Logic (DEL)}

We recap here the new kind of model for epistemic logic based on 
chromatic simplicial complexes, introduced in~\cite{gandalf2018}. 
The geometric nature of simplicial complexes allows us to consider higher-dimensional topological
properties of our models, and investigate their meaning in terms of knowledge.
The idea of using simplicial complexes comes from distributed computability~\cite{HerlihyKR:2013,kozlov:2007}.
After describing simplicial models, we explain how to use them in DEL.

\paragraph{Syntax.}
Let $\Atom$ be a countable set of atomic propositions and $\Agents$ a finite set of agents.
The language $\mathcal{L}_K$ is generated by the following BNF
grammar:
\[
\varphi ::= p \mid \neg\varphi \mid (\varphi \land \varphi) \mid
K_a\varphi \qquad p \in \Atom,\ a \in \Agents
\]
In the following, we work with $n+1$ agents, and write $\Agents = \{ a_0, \ldots, a_n \}$.

\paragraph{Semantics.}
The usual semantics for multi-agent epistemic logic is based on Kripke frames.
Here, we use the recent notion of model based on simplicial complexes, which makes
explicit the topological information of Kripke frames. The precise relationship
between the usual Kripke models and our simplicial models is studied thoroughly in~\cite{gandalf2018}.

\begin{definition} [Simplicial complex \cite{kozlov:2007}]
A \emph{simplicial complex} $\langle V, M \rangle$ is given by a set $V$ of \emph{vertices} and a family $M$ of non-empty finite subsets of $V$ called \emph{simplices}, such that
for all $X \in M$, $Y \subseteq X$ implies $Y \in M$. We say that~$Y$ is a \emph{face} of $X$.
\end{definition}

Usually, the set of vertices is implicit and we simply refer to a simplicial complex as $M$. We write $\vertices{M}$ for the set of vertices of $M$. 
A vertex $v \in \vertices{M}$ is identified with the singleton $\{v\} \in M$. 
Elements of $M$ are called \emph{simplices}, and those which are maximal w.r.t.\ inclusion are \emph{facets} (or \emph{worlds}),  the set of which  is denoted by $\facets{M}$. 
 The \emph{dimension} of a simplex $X \in M$ is  $|X|-1$. 
 A simplicial complex $M$ is \emph{pure} if all its facets are of the same dimension $n$. In this case,
 we say $M$ is of dimension $n$.
Given a finite set $\Agents$ of agents (that we will represent as colors), a \emph{chromatic simplicial complex} $\la M, \chi \ra$ consists of a simplicial complex $M$ and a coloring map $\chi : \vertices{M} \to \Agents$, such that for all $X \in M$, all the vertices of $X$ have distinct colors.

\begin{definition}[Simplicial map]
\label{def:morphism_of_complexes}
Let $C$ and $D$ be two simplicial complexes. A \emph{simplicial map} 
$f: C \rightarrow D$ 
maps the vertices of $C$ to vertices of $D$, such that if~$X$ is a simplex of $C$, $f(X)$ 
is a simplex of $D$.
A \emph{chromatic simplicial map} between two chromatic simplicial complexes is a simplicial map that preserves colors.
\end{definition}

%

For technical reasons, we restrict to models where all the atomic propositions
are saying something about some local value held by one particular agent. All the examples 
that we are interested in will fit in that framework. Let $\Values$ be some 
countable set of values, and $\Atom = \{ p_{a,x} \mid a \in \Agents, x \in \Values \}$ 
be the set of \emph{atomic propositions}. Intuitively, $p_{a,x}$ is true if agent 
$a$ holds the value $x$.
We write $\Atom_a$ for the atomic propositions concerning agent $a$.


A \emph{simplicial model} $\cM = \la C, \chi, \ell \ra$ consists of a pure chromatic simplicial 
complex $\la C, \chi \ra$ of dimension $n$, and a labeling $\ell : \vertices{C} \to \mathscr{P}(\Atom)$
that associates with each vertex $v \in \vertices{C}$ a set of atomic propositions concerning agent $\chi(v)$, i.e., such that $\ell(v) \subseteq \Atom_{\chi(v)}$.
Given a facet $X = \{v_0, \ldots, v_n\} \in C$, we write $\ell(X) = \bigcup_{i=0}^n \ell(v_i)$.
A \emph{morphism} of simplicial models $f : \cM \to \cM'$ is a chromatic simplicial map that preserves the labeling: $\ell'(f(v)) = \ell(v)$ (and $\chi$).

\begin{definition}\label{semantFormulas}
We define the truth of a formula $\varphi$ in some epistemic state $(\cM,X)$ with $\cM=\la C, \chi, \ell \ra$ a simplicial model, $X \in \facets{C}$ a facet of $C$ and
$\varphi \in \Lcal_K(\Agents,\Atom)$.
The satisfaction relation, determining when a formula is true in an
epistemic state, is defined as:

\begin{tabular}{lcl}
$\cM,X \models p$ & \quad if \quad & $p \in \ell(X)$\\
$\cM,X \models \neg \varphi$ & \quad if \quad & $\cM,X \not\models \varphi$\\
$\cM,X \models \varphi \wedge \psi$ & \quad if \quad & $\cM,X \models \varphi
\mbox{ and } \cM,X \models \psi$\\
$\cM,X \models K_a \varphi$ & \quad if \quad & $\mbox{for all } Y \in \cF(C), a \in \chi(X \cap Y) \mbox{ implies } \cM,Y \models \varphi$\\
\end{tabular}
\end{definition}

It is not hard to see that this definition of truth agrees with the usual one on Kripke models (see~\cite{gandalf2018}). 


\subsubsection{DEL and its topological semantics.}

DEL is the study of modal logics of model change~\cite{sep-dynamic-epistemic,DEL:2007}.
A modal logic studied in DEL is obtained
by using action models~\cite{baltagMS:98}, which are relational structures that can be used to describe a variety of communication actions.

\paragraph{Syntax.}
We extend the syntax of epistemic logic with one more construction:
\[
\varphi ::= p \mid \neg\varphi \mid (\varphi \land \varphi) \mid
K_a\varphi \mid [\alpha]\varphi \qquad p \in \Atom,\ a \in \Agents
\]
Intuitively, $[\alpha]\varphi$ means that $\varphi$ is true after some 
\emph{action} $\alpha$ has occurred.
An action can be thought of as an 
announcement made by the environment,
which is not necessarily public, in the sense that
not all  agents receive these announcements.
The semantics of this new operator should be understood as follows:
 \[ \cM,X \models [\alpha] \varphi \quad \text{if} \quad \cM[\alpha],X[\alpha] \models \varphi\]
i.e., the formula $[\alpha] \varphi$ is true in some world $X$ of $\cM$ whenever $\phi$ is true in some new model $\cM[\alpha]$, where the knowledge of each agent has been modified according to the action $\alpha$.
To define formally what an action is, we first need to introduce the notion of \emph{action model}.
An action model describes all the possible actions that might happen, as well
as how they affect the different agents.

\paragraph{A simplicial complex version of DEL.}

An \emph{action model} is a structure $\cA = \la T,\sim,\pre \ra$,
where~$T$ is a domain of \emph{actions}, such that for
each $a \in \Agents$, $\sim_a$ is an equivalence relation on $T$, and
$\pre : T \to \Lcal_\cK$ is a  function that assigns a
\emph{precondition} formula $\pre(t)$ to each $t \in T$.
An action model is \emph{proper} if for any two different actions $t, t' \in T$, there is an agent $a \in \Agents$ who can distinguish between them, i.e., $t \not \sim_a t'$.

Given a simplicial model $\cM = \langle C, \chi, \ell \rangle$ and an action model $\cA = \la T, \sim, \pre \ra$, we define the \emph{product update simplicial model} $\cM[\cA] = \la C[\cA], \chi[\cA], \ell[\cA]  \ra$ as follows. 
Intuitively, the facets of $C[\cA]$ should correspond to pairs $(X,t)$ where $X \in C$ is a world of $\cM$ and $t \in T$ is an action of $\cA$, such that $\cM,X \models \pre(t)$.
Moreover, two such facets $(X,t)$ and $(Y,t')$ should be glued along their $a$-colored vertex whenever $a \in \chi(X \cap Y)$ and $t \sim_a t'$.
Formally, the vertices of $C[\cA]$ are pairs $(v,E)$ where $v \in \cV(C)$ is a vertex of $C$; $E$ is an equivalence class of  $\sim_{\chi(v)}$; and $v$ belongs to some facet $X \in C$ such that there exists $t \in E$ such that $\cM,X \models \pre(t)$.
Such a vertex keeps the color and labeling of its first component: $\chi[\cA](v,E) = \chi(v)$ and $\ell[\cA](v,E) = \ell(v)$.

Given a product update simplicial model $\cM[\cA] = \la C[\cA], \chi[\cA], \ell[\cA] \ra$ as above, 
	one can naturally enrich it	by extending the set of atomic propositions in order 
	to capture the equivalence class of~$\sim_{\chi(v)}$ on each vertex~$v$.
The extended set of atomic propositions would then be
	$\ExtAtom=\Atom\cup \{p_E \mid E\in T/\!\sim_a,\, a\in \Agents\}$, 
	where $T/\!\sim_a$ denotes the set of all equivalence classes of~$\sim_a$.
In that case, the \emph{extended product update model} is
	$\ExtProd{\M}{\A} = \langle C[\cA], \chi[\cA], \extended{\ell}[\cA]\rangle$,
	that differs from~$\cM$ only in labeling.
Namely, the enriched labeling $\extended\ell[\cA]$ 
	maps each vertex $(v,E)\in C[\cA]$ into the set of atomic propositions 
	$\extended\ell[\cA]((v,E))= \ell(v)\cup \{p_E\}$.
	On this extended model $\ExtProd{\cM}{\cA}$, we can interpret formulas saying something not only about the atomic propositions of $\cM$, but also about the actions that may have occurred.

\medskip

In the next section, we describe a particular action model of interest,
the one corresponding to the layered message-passing model described in Section~\ref{sec:dcomp}.

\subsection{The layered message-passing action model}	
\label{sec:dcomp}

This section starts with an overview of the \emph{layered message-passing model} for two 
agents, or \emph{processes} as they are called in distributed computing.
More details about this model can be found in~\cite{HerlihyKR:2013}.
This model is known to be equivalent to the well-studied \emph{read/write wait-free model}, in the sense that it solves the same set of tasks.
When there are only two processes involved in the computation, which is what we want to study in this article, the layered message-passing model is easier to understand.
Here, we formalize this model as an action model; a more usual presentation can be found in Appendix~\ref{app:layered_model}, as well as a proof of equivalence between the two.

\subsubsection{The layered message-passing model.}
\label{sec:DCmodel}
Let the processes be $B,W$, to draw them in the pictures with colors black and white.
In the \emph{layered message-passing} model, computation is synchronous:
$B$ and $W$ take steps at the same time.
We will call each such step a \emph{layer}.
In each layer,
$B$ and $W$ both send a message to each other, 
where at most one message may fail to arrive,
implying that either one or two messages will be received.
This is a \emph{full information} model, in the sense that each time a process
sends a message, the message consists of its local state (i.e., all the information currently known to the process), and each time it receives a message,
it appends it to its own local state (remembers everything).
A protocol is defined by the number~$N$ of layers the processes execute. Then, each process should produce an output value based on its state at the end of the last layer. A decision function $\delta$ specifies the output value of each process at the end of the last layer. 

Given an initial state, an execution can be specified by a sequence of $N$ symbols over the alphabet $\set{\bot,B,W}$,
meaning that, if the $i$-th symbol in the sequence is $\bot$ then in the $i$-th layer both messages arrived,
and if the $i$-th symbol is $B$ (resp. $W$) then only $B$'s message failed to arrive (resp. $W$) in the $i$-th layer.
As an example, $\bot B W$ corresponds to an execution in which both processes have received each others message at layer one, then~$B$ received the message from $W$ but $W$ did not receive the message from~$B$ at layer two, and finally at layer
three, $W$ received the message from~$B$ but~$B$ did not receive the message from~$W$.

For example, there are three $1$-layer executions, namely $\bot$, $B$ and $W$, but from the point of view of process $B$, there are two distinguished cases: (i) either it did not receive a message, in which case it knows for sure that the execution that occurred was $W$, or (ii) it did receive a message from $W$, in which case the execution could have been either $B$ or $\bot$.
Thus, for the black process executions $B$ and $\bot$ are indistinguishable.

\subsubsection{The layered message-passing model as an action model.}
\label{sec:layered_model}

Consider the situation where the agents $\Agents =\set{B,W}$ each
start in an initial global state, defined by input values given to each agent.
The values are local, in the sense that each agent knows its own initial value,
but not necessarily the values given to other agents. 
The agents communicate to each other 
via the layered message-passing model described above.
The layered message-passing action model described next is equivalent 
to the immediate snapshot action model of~\cite{gandalf2018} in the case of two processes. 

Let $\Vin$ be an arbitrary domain of \emph{input values}, and take the following set of atomic propositions $\Atom = \{ \inputprop{a}{x} \mid a \in \Agents,\, x \in \Vin \}$.
Consider a simplicial model $\cI = \la I, \chi, \ell \ra$ called the \emph{input simplicial model}.
Moreover, we assume that for each vertex $v \in \cV(I)$, corresponding to some agent $a = \chi(v)$, the labeling $\ell(v) \subseteq \Atom_{a}$ is a singleton,
assigning to the agent $a$ its private input value.
A facet $X \in \cF(I)$ represents
a possible initial configuration, where each agent has been given an input value.

The  {action model}  $\IS= \la T,\sim,\pre \ra$ corresponding to $N$ layers is defined as follows.
Let $L_N$ be the set of all sequences of $N$ symbols over the alphabet $\set{\bot,B,W}$.
Then, we take~$T=L_N\times \cF(I)$. An action $(\alpha,X)$, where $\alpha\in L_N$ and $X \in \cF(I)$ 
represents a possible execution starting in the  initial configuration~$X$.
We write $X_a$ for the input value assigned to agent $a$ in the input simplex $X$.
Then, $\pre : T \to \Lcal_\cK$  assigns to each $(\alpha,X) \in T$ a 
{precondition} formula $\pre(\alpha,X)$ which holds exactly in $X$ 
(formally, we take $\pre(\alpha,X) = \bigwedge_{a \in \Agents} \inputprop{a}{X_a}$).
To define the indistinguishability relation $\sim_a$, we proceed by induction on $N$.
For $N=0$, we define $(\varnothing, X) \sim_a (\varnothing, Y)$ when $X_a = Y_a$, since process $a$ only sees its own local state.
Now assume that the indistinguishability relation of $\IS$ has been defined, we define $\sim_a$ on $\mathcal{MP}_{N+1}$ as follows.
Let  $\alpha,\beta\in L_N$ and $p,q\in\set{\bot,B,W}$. 
We define $(\alpha \cdot p, X) \sim_B (\beta \cdot q, Y)$ if either:
\begin{enumerate}[label={$(\roman*)$}]
\item $p = q = W$ and $(\alpha, X) \sim_B (\beta, Y)$, or
\item $p,q \in \{\bot, B\}$ and $X=Y$ and $\alpha = \beta$, 
\end{enumerate}
and similarly for $\sim_W$, with the role of $B$ and $W$ reversed.
Intuitively, either (i) no message was received, and the uncertainty from the previous layers remain; or (ii) a message was received, and the process $B$ can see the whole history, except that it does not know whether the last layer was $B$ or $\bot$.


%

To see what the effect of this action model is, let us start with an input model~$\cI$ with only one input configuration $X$ (input values have been omitted).
\begin{center}
\begin{tikzpicture}[auto, scale=2, whitevertex/.style={draw=black,thick,circle,fill=white,inner sep=2pt,minimum size=1em}, blackvertex/.style={draw=black,thick,circle,fill=black,inner sep=2pt,minimum size=1em,text=white}]
\node[blackvertex] (b0) at (0,1) {};
\node[whitevertex] (w1) at (1,1) {};
\path	  (b0) edge  node[above]  {}   (w1);
\end{tikzpicture}
\end{center}
After one layer of the message passing model, we get the following model $\cI[\mathcal{MP}_1]$:
\begin{center}
\begin{tikzpicture}[auto, scale=2, whitevertex/.style={draw=black,thick,circle,fill=white,inner sep=2pt,minimum size=1em}, blackvertex/.style={draw=black,thick,circle,fill=black,inner sep=2pt,minimum size=1em,text=white}]
\node[blackvertex] (b0) at (0,1) {};
\node[whitevertex] (w1) at (1,1) {};
\node[blackvertex] (b2) at (2,1) {};
\node[whitevertex] (w3) at (3,1) {};
\path	  (b0) edge  node[above]  {$W$}   (w1)
          (w1) edge node[above]  {$\bot$} (b2)
          (b2) edge node[above]  {$B$} (w3);
\end{tikzpicture}
\end{center}
After a second layer, we get $\cI[\mathcal{MP}_2]$:
\begin{center}
\begin{tikzpicture}[auto, scale=1.3, whitevertex/.style={draw=black,thick,circle,fill=white,inner sep=2pt,minimum size=1em}, blackvertex/.style={draw=black,thick,circle,fill=black,inner sep=2pt,minimum size=1em,text=white}]

\node[blackvertex] (b0) at (0,1) {};
\node[whitevertex] (w1) at (1,1) {};
\node[blackvertex] (b2) at (2,1) {};
\node[whitevertex] (w3) at (3,1) {};

\node[blackvertex] (b4) at (4,1) {};
\node[whitevertex] (w5) at (5,1) {};
\node[blackvertex] (b6) at (6,1) {};

\node[whitevertex] (w7) at (7,1) {};
\node[blackvertex] (b8) at (8,1) {};
\node[whitevertex] (w9) at (9,1) {};

\path	  (b0) edge  node[above]  {$WW$}   (w1) (w1) edge node[above]  {$W\bot$} (b2) (b2) edge node[above]  {$WB$} (w3) (w3)
edge node[above]  {$\bot B$}  (b4) (b4)  edge node[above]  {$\bot \bot$} (w5) (w5) edge node[above]  {$\bot W$} (b6) (b6) edge node[above]  {$B W$} (w7) (w7)  edge node[above]  {$B \bot$} (b8) (b8) edge node[above]  {$B B$} (w9);

\end{tikzpicture}
\end{center}

The remarkable property of this action model, is that it preserves the topology of the input model.
This is a well-known fact in distributed computing~\cite{HerlihyKR:2013}, reformulated here in terms of DEL.

\begin{theorem}
\label{th:mainTopInvIS}
Let  $\cI = \la I, \chi, \ell \ra$ be an {input model}, and 
$\IS= \la T,\sim,\pre \ra$ be the $N$-layer action model.
Then, the product update 
simplicial model $\cI [\IS]$ 
is a subdivision of $\cI$, where each edge is subdivided into $3^N$ edges.
\end{theorem}

\subsection{Outline of impossibility proofs}
\label{sec:proofmethod}

We now describe how the set up of~\cite{gandalf2018} is used to prove impossibility results in distributed computing.
It is closely related to the usual topological approach to distributed computability~\cite{HerlihyKR:2013}, except that the input complex, output complex and protocol complex are now viewed as simplicial models for epistemic logic.
By interpreting epistemic logic formulas on those structures, we can understand the epistemic content of the abstract topological arguments for unsolvability.
For example, when the usual topological proof would claim that consensus is not solvable because the protocol complex is connected, our DEL framework allows us to say that the reason for impossibility is that the processes did not reach common knowledge of the set of input values. This particular example, among others, is treated in depth in~\cite{gandalf2018}.

As in the previous section, we fix an input simplicial model $\cI = \la I, \chi, \ell \ra$.
A \emph{task} for $\cI$ is an action model
   ${\mathcal{T}}=\la T, \sim, \pre\ra$
for agents $\Agents$, where each action $t \in T$ consists of a function $t : \Agents \to \Vout$,
 where $\Vout$ is an arbitrary
 domain of \emph{output values}.
 Such an action is interpreted as an assignment of an output value for each agent.
 Each such $t$ has 
a precondition that is true in one or more facets of $\cI$,
interpreted as ``if the input configuration is a facet in which $\pre(t)$ holds,
and every agent $a \in \Agents$ decides the value $t(a)$, then this is a valid execution''.
\longversion{
  The indistinguishability relation is defined as $t \sim_a t'$ when $t(a) = t'(a)$.
}

\begin{definition}
\label{def:solving_task}
Task $\cT$ is solvable in $\cI[\IS]$ if there exists a morphism $\delta: \cI[\IS]\rightarrow \cI[{\cT}]$ such that
$\pi\, \circ\, \delta=\pi$, i.e., the diagram of simplicial complexes below commutes.
\end{definition}

\begin{wrapfigure}[8]{r}{0.2\textwidth}
\begin{center}
\vspace{-35pt}
\begin{tikzpicture}
  \node (s) {$\I[\IS]$};
  \node (xy) [below=2 of s] {$\I[\T]$}; 
  \node (x) [left=of xy] {$\I$};
  \draw[<-] (x) to node [sloped, above] {$\pi$} (s);
  \draw[->, dashed, right] (s) to node {$\delta$} (xy);
  \draw[->] (xy) to node [below] {$\pi$} (x);
\end{tikzpicture}
\end{center}
\end{wrapfigure}
In the above definition, the two maps denoted as $\pi : \cI[\IS] \to \cI$ and $\pi : \cI[\cT] \to \cI$ are simply projections on the first component.
The intuition behind this definition is the following.
A facet $X$  in $\cI[\IS]$ corresponds to a pair $(i,act)$,
where ${i \in \cF(\cI)}$ represents input value assignments to all agents, 
and $act\in \IS$ represents an action,
codifying the communication exchanges that took place.
The morphism~$\delta$ takes~$X$ to a facet $\delta(X) = (i,t)$ of $\cI[\cT]$,
where $t \in \cT$ is assignment of decision values that the agents will
choose in the situation $X$.

Moreover, $\pre(t)$ holds in $i$, meaning that $t$ corresponds to valid decision values for input $i$.
The commutativity of the diagram expresses the fact that both $X$ and $\delta(X)$ correspond to the same input assignment $i$.
Now, consider a single vertex $v \in X$ with $\chi(v) = a \in \Agents$. Then, agent $a$
decides its value solely according to its knowledge in $\cI[\IS]$: if another
facet $X'$ contains $v$, then $\delta(v) \in \delta(X) \cap \delta(X')$, meaning
that $a$ has to decide the same value in both situations.

To prove impossibility results, our goal is thus to show that no such map $\delta$ can exist. To do so, we rely on the following lemma, which is a reformulation in the simplicial setting of a classic result of modal logics.

\begin{lemma} [\cite{gandalf2018}]
\label{lem:gain_knowledge}
Consider simplicial models $\cM=\la C, \chi, \ell \ra$ and
$\cM' = \la C', \chi', \ell' \ra$, and a  morphism $f : \cM \to \cM'$.
Let $X \in \cF(C)$ be a facet of $\cM$, $a$ an agent, and $\varphi$ a formula 
which does not contain negations except, possibly, in front of atomic propositions.
Then, $\cM',f(X) \models \varphi$ implies $\cM,X \models \varphi$.
\end{lemma}

To prove that a task $\cT$ is not solvable in $\IS$, our usual proof method goes like this. Assume $\delta \colon \cI[\IS] \to \cI[\cT]$ exists, then:
\begin{enumerate}
\item Pick a well-chosen positive epistemic logic formula $\phi$,
\item Show that $\phi$ is true in every world of $\cI[\cT]$,
\item Show that there exists a world $X$ of $\cI[\IS]$ where $\phi$ is false,
\item By Lemma~\ref{lem:gain_knowledge}, since $\phi$ is true in $\delta(X)$ then it must also be true in~$X$, which is a contradiction with the previous point.
\end{enumerate}
This kind of proof is interesting because it explains the reason why the task is not solvable. The formula $\phi$ represents some amount of knowledge which the processes must acquire in order to solve the task. If~$\phi$ is given, the difficult part of the proof is usually the third point: finding a world $X$ in the protocol complex where the processes did not manage to obtain the required amount of knowledge.
The existence of this world can be proved using theorems of combinatorial topology, such as Sperner's Lemma; see~\cite{gandalf2018} for such examples.

\section{Equality negation task for two processes}
\label{sec:negationTask}

The equality negation task has been introduced 
		in~\cite{DBLP:journals/siamcomp/LoH00}, and further studied in~\cite{DISC19}.
In this section, we will be interested only in the case of two processes.
Each process starts with an input value in the set $\{0, 1, 2\}$, and has
to irrevocably decide on a value~$0$ or~$1$, such that the decisions of the two
processes are the same if and only if their input values are different.
In~\cite{DBLP:journals/siamcomp/LoH00} it has been proved that the 
equality negation task is unsolvable for two processes in a wait-free model
using only registers. We reproduce this proof and give a more direct proof
in Appendix~\ref{sec:app:negation}.
In this section we  analyze this task using our DEL framework and use it
to prove the unsolvability of the equality negation task.

\subsection{DEL analysis of the task}
\label{sec:DEL_analysis}

Let $\Agents = \{\B,\W\}$ be the two agents (or processes).
In the pictures, process $\B$ will be associated to {black} vertices,
and process $\W$ with  {white} vertices.
The atomic propositions are of the form $\inputprop{p}{i}$, for $p \in \Agents$ and $i \in \{0, 1, 2\}$, meaning that process $p$ has input value $i$.
The input model is $\I = \langle I, \chi, \ell \rangle$ where:
\begin{itemize}
\item $I$ is the simplicial complex whose set of vertices is $\vertices{I} = \Agents \times \{0, 1, 2\}$, and whose facets are of the form $\{(B,i), (W,j)\}$ for all $i, j$.
\item The coloring $\chi : \vertices{I} \to \Agents$ is the first projection $\chi(p,i) = p$.
\item $\ell(p,i) = \{\inputprop{p}{i}\}$.
\end{itemize}

The input model $\I$ is represented below. In the picture, a vertex
${(p,i) \in \vertices{I}}$ is represented as a vertex of color $p$ with
value $i$.

\begin{center}
\begin{tikzpicture}[auto, scale=1.5, whitevertex/.style={draw=black,thick,circle,fill=white,inner sep=2pt,minimum size=1em}, blackvertex/.style={draw=black,thick,circle,fill=black,inner sep=2pt,minimum size=1em,text=white}]

\node[whitevertex] (w0) at (0,0) {$0$};
\node[whitevertex] (w1) at (1,1) {$1$};
\node[whitevertex] (w2) at (2,0) {$2$};
\node[blackvertex] (b0) at (0,1) {$0$};
\node[blackvertex] (b1) at (1,0) {$1$};
\node[blackvertex] (b2) at (2,1) {$2$};
\path (w0) edge (b0) edge (b1) edge (b2)
	  (w1) edge (b0) edge (b1) edge (b2)
      (w2) edge (b0) edge (b1) edge (b2);
\end{tikzpicture}
\end{center}

We now define the action model $\T = \langle T, \sim, \pre \rangle$ that
specifies the task.
Since the only possible outputs are~$0$ and~$1$, there are four possible
actions: $T = \{0,1\}^2$, where by convention the first component is the
decision of $\B$, and the second component is the decision of $\W$.
Thus, two actions $(d_\B,d_\W) \sim_\B (d_\B',d_\W')$ in~$T$ are
indistinguishable by $\B$ when  $d_\B=d_\B'$, and similarly for $\W$.
Finally, the precondition $\pre(d_\B,d_\W)$ specifies the task as expected:
if $d_\B=d_\W$ then $\pre(d_\B,d_\W)$ is true exactly in the simplices of $\I$ 
which have different input values, and otherwise in all the simplices which
have identical inputs.

The \emph{output model} is obtained as the product update model
$\O = \I[\T] = \langle O, \chi_\O, \ell_\O \rangle$.
By definition, the vertices of $O$ are of the form $(p,i,E)$, where
$(p,i) \in \vertices{I}$ is a vertex of $\I$, and $E$ is an equivalence class
of $\sim_p$.
But note that~$\sim_p$ has only two equivalence classes, depending on the
decision value ($0$ or~$1$) of process~$p$. So a vertex of $O$ can be written
as $(p,i,d)$, meaning intuitively that process $p$ started with input~$i$ and decided value~$d$.
The facets of $O$ are of the form $\{(\B,i,d_\B), (\W,j,d_\W)\}$ where either
$i=j$ and $d_\B\neq d_\W$, or $i\neq j$ and $d_\B=d_\W$.
The coloring $\chi_\O$ and labeling $\ell_\O$ behave the same as in $\I$.

The output model for the equality negation task is depicted below.
Decision values do not appear explicitly on the picture, but notice how the
vertices are arranged as a rectangular cuboid: the vertices on the front face
have decision value~$0$, and those on the rear face decide~$1$.

\begin{center}
\begin{tikzpicture}[auto, scale=1.8, whitevertex/.style={draw=black,thick,circle,fill=white,inner sep=2pt,minimum size=1em}, blackvertex/.style={draw=black,thick,circle,fill=black,inner sep=2pt,minimum size=1em,text=white}]

\node[whitevertex] (w00) at (0,0) {$0$};
\node[whitevertex] (w10) at (1,1) {$1$};
\node[whitevertex] (w20) at (2,0) {$2$};
\node[blackvertex] (b20) at (0,1) {$2$};
\node[blackvertex] (b10) at (1,0) {$1$};
\node[blackvertex] (b00) at (2,1) {$0$};
\node[whitevertex] (w01) at (2.4,1.4) {$0$};
\node[whitevertex] (w11) at (1.4,0.4) {$1$};
\node[whitevertex] (w21) at (0.4,1.4) {$2$};
\node[blackvertex] (b01) at (0.4,0.4) {$0$};
\node[blackvertex] (b11) at (1.4,1.4) {$1$};
\node[blackvertex] (b21) at (2.4,0.4) {$2$};
\path (w00) edge (b01) edge (b10) edge (b20)
	  (w10) edge (b00) edge (b11) edge (b20)
      (w20) edge (b00) edge (b10) edge (b21)
      (w01) edge (b00) edge (b11) edge (b21)
	  (w11) edge (b01) edge (b10) edge (b21)
      (w21) edge (b01) edge (b11) edge (b20);
\end{tikzpicture}
\end{center}

We want to prove that this task is not solvable when the processes communicate
through $N$ layers of our message passing model, no matter how large $N$ is selected. 
Thus, what we need to show is that there is no morphism $\delta : \I[\IS] \to \O$ that makes the diagram of Definition~\ref{def:solving_task} commute. 
%
In Section~\ref{sec:bisimulation_and_limits}, we will show that the general proof method described in Section~\ref{sec:proofmethod} actually fails. In Section~\ref{sec:extendedDEL}, we extend the expressivity of our logic in order to obtain an epistemic proof that $\delta$ does not exist. 


\subsection{Bisimulation and limits of the DEL framework}
\label{sec:bisimulation_and_limits}

We would like to use the proof method described in
Section~\ref{sec:proofmethod} to find a logical obstruction showing that the
morphism $\delta$ cannot exist, through a formula $\phi$.
%
 To show that, in fact, there is no such suitable formula
$\phi$, we first need to define bisimulations for simplicial models.

\begin{definition}[Bisimulation]\label{def:bisimulation}
Let $\M = \langle M, \chi, \ell \rangle$ and $\M' = \langle M', \chi', \ell'
\rangle$ be two simplicial models. 
Relation ${R \subseteq \facets{M} \times \facets{M'}}$ is a 
	\emph{bisimulation} between $\M$ and $\M'$ if the following conditions hold:
\begin{enumerate}[label=(\roman*)]
\item \label{atoms}
  If $X \mathrel{R} X'$ then $\ell(X) = \ell'(X')$.
\item \label{forth}
  For all 
  $a \in \Agents$,
  if $X \mathrel{R} X'$ and ${a \in \chi(X \cap Y)}$,
  then there exists $Y' \in \facets{M'}$ such that $Y \mathrel{R} Y'$ and
  $a \in \chi'(X' \cap Y')$.
\item \label{back}
  For all 
  $a \in \Agents$,
  if $X \mathrel{R} X'$ and ${a \in \chi'(X' \cap Y')}$,
  then there exists ${Y \in \facets{M}}$ such that $Y \mathrel{R} Y'$ and
  $a \in \chi(X \cap Y)$.
\end{enumerate}
When $R$ is a bisimulation and $X \mathrel{R} X'$, we say that $X$ and $X'$ are \emph{bisimilar}.
\end{definition}

The next lemma states that two bisimilar worlds satisfy exactly the
same formulae. This is a well-known fact in the context of Kripke models. 
The same results holds for bisimulations between simplicial models.

\begin{lemma}
\label{lem:bisimulation_truth}
Let $R$ be a bisimulation between $\M$ and $\M'$.
Then for all facets $X, X'$ such that $X \mathrel{R} X'$, and for every
epistemic logic formula $\phi$,
  \[\M,X \models \phi \qquad \text{iff} \qquad \M',X' \models \phi\]
\end{lemma}

\begin{proof}
In order to simplify our reasoning, we only demonstrate the following:
\begin{equation}\label{eq:bisim.invar}
\M,X \models \phi \qquad \mbox{implies} \qquad \M',X' \models \phi.
\end{equation}
This is enough to claim the equivalence, as the relation~$\mathrel{R}$ is symmetric.
Our proof is by induction on $\phi$.

Let $X\in\facets M$ and $X'\in\facets{M'}$ be facets with $X \mathrel{R} X'$.

\textsc{Base case.}
Let us first consider the case when~$\phi=p\in\Atom$. 
By Definition~\ref{semantFormulas} we have that $\M,X \models p$ if and only if $p \in \ell(X)$.
As $X \mathrel{R} X'$, we know that $\ell(X)=\ell'(X')$ by Definition~\ref{def:bisimulation}\ref{atoms}.
That is, we have that $p\in\ell'(X')$, which, by Definition~\ref{semantFormulas},
	holds if and only if $\M',X' \models p$.

\textsc{Induction hypothesis.}
Let us now assume that the claim it true for formulas~$\psi$ and~$\psi_1$, that is, 
	assume that for any $X\mathrel{R}X'$ the following holds:
\begin{equation}\label{eq:ih.psi}
\M,X \models \psi \qquad \text{iff} \qquad \M',X' \models \psi
\end{equation}
\begin{equation}\label{eq:ih.psi1}
\M,X \models \psi_1 \qquad \text{iff} \qquad \M',X' \models \psi_1
\end{equation}

\textsc{Induction step.}
We prove that (\ref{eq:bisim.invar}) holds for 
	(i) $\phi = \neg \psi$, 
	(ii) $\phi = \psi \wedge \psi_1$, and 
	(iii) $\phi = K_a \psi$, for an $a\in\Agents$.

(i) Let $\M,X \models \neg \psi$. 
Then by Definition~\ref{semantFormulas} we have that $\M,X \not\models \psi$, and 
	thus by Induction Hypothesis~(\ref{eq:ih.psi}) it also holds that $\M',X' \not\models \psi$.
Again by Definition~\ref{semantFormulas} we have that $\M',X' \models \neg\psi$.

(ii) If $\M,X\models\psi\wedge\psi_1$, this means that $\M,X\models\psi$
	and $\M,X\models\psi_1$. By Induction Hypothesis we then have that 
	$\M',X'\models\psi$	and $\M',X'\models\psi_1$, which is by definition
	equivalent to $\M,X\models\psi\wedge\psi_1$.
	
(iii) To this end, let $\M,X \models K_a \psi$ for an $a\in \Agents$.
By Definition~\ref{semantFormulas} we have the following:
\begin{equation}\label{eq:Kdef}
\forall Y \in \cF(M), \quad a \in \chi(X \cap Y) \;\; \rightarrow \;\; \M,Y \models \psi
\end{equation}
	%
We want to show that $\M',X' \models K_a \psi$, that is, for every $Y' \in \cF(M')$ such that 
	$a \in \chi'(X' \cap Y')$, it holds that $\M',Y' \models \psi$.
Let $Y'$ be an arbitrary facet of $M'$ such that $a \in \chi'(X' \cap Y')$.
As $X \mathrel{R} X'$, by Definition~\ref{def:bisimulation}\ref{back} we know that there exists a $Y\in\cF(M)$
	with $a \in \chi(X \cap Y)$ and $Y\mathrel{R} Y'$.
As  $a \in \chi(X \cap Y)$, from the condition~(\ref{eq:Kdef}) we obtain that $\M,Y\models\psi$.
Finally, as $Y\mathrel{R} Y'$, we can apply our Induction Hypothesis, 
	and conclude that $\M',Y' \models \psi$.
\qed
\end{proof}

We now come back to the equality negation task for two processes.
As it turns out, there is a bisimulation between the input and output models.

\begin{lemma}\label{lem:IbisimO}
Let~$\I$ and~$\O$ be the input and output models of the 
	equality negation task, respectively, and let~$\pi$ be the projection map $\pi : \O \to \I$.
The relation $R = \{(\pi(X), X) \mid X \in \facets{\O}\} \subseteq \I \times \O$
	is a bisimulation between~$\I$ and~$\O$.
\end{lemma}

\begin{proof}
The first condition of Definition~\ref{def:bisimulation} is trivially fulfilled.

Let us check that condition \ref{forth} is verified.
Let $X$ and $X'$ be facets of $\I$ and $\O$ respectively, such that $X \mathrel{R} X'$.
Thus, we have $X = \{(\B,i), (\W,j)\}$ and $X' = \{(\B,i,d_\B), (\W,j,d_\W)\}$,
for some $i,j,d_\B,d_\W$. 
Now let $a \in \Agents$ (w.l.o.g., let us pick $a=\B$), and assume that there
is some $Y \in \facets{\I}$ such that $\B \in \chi(X \cap Y)$.
So, $Y$ can be written as $Y = \{(\B,i), (\W,j')\}$ for some $j'$.
We now need to find a facet $Y'$ of $\O$ that shares a $\B$-colored
vertex with $X'$, and whose projection $\pi(Y')$ is $Y$.
Thus, $Y'$ should be of the form $Y' = \{(\B,i,d_\B), (\W,j',d_\W')\}$,
for some $d_\W'$, such that $i = j' \iff d_\B \neq d_W'$ .
But whatever the values of $i, j', d_\B$ are, we can always choose a suitable
$d_\W'$. This concludes the proof.

The third condition \ref{back} is checked similarly.
\qed
\end{proof}

\medskip

We can finally use Lemma~\ref{lem:bisimulation_truth} to show that no formula
$\phi$ will allow us to prove the unsolvability of the equality negation task.

\begin{lemma}
\label{lem:no_phi}
For the equality negation task, let $X$ be a facet of $\I[\IS]$ and let~$Y$ 
	be a facet of $\O$ such that $\pi(X) = \pi(Y)$. 
Then for every positive formula $\phi$ we have the following:
	if $\O,Y \models \phi$ then $\I[\IS],X \models \phi$.
\end{lemma}
\begin{proof}
Let $\phi$ be a positive formula and assume $\O,Y \models \phi$.
Since we have shown in Lemma~\ref{lem:IbisimO} that $\pi(Y)$ and $Y$ are bisimilar,
by Lemma~\ref{lem:bisimulation_truth}, we have $\I,\pi(Y) \models \phi$.
Since that $\pi(Y) = \pi(X)$, 
by Lemma~\ref{lem:gain_knowledge} we obtain ${\I[\IS],X \models \phi}$.
\qed
\end{proof}

In the above lemma, the world $Y$ should be thought of as a candidate
for~$\delta(X)$. The condition $\pi(X) = \pi(Y)$ comes from the commutative
diagram of Definition~\ref{def:solving_task}. Thus, Lemma~\ref{lem:no_phi} says that we
will never find a formula $\phi$ which is true in $\delta(X)$ but false in $X$.

\longversion{
  \medskip
  
  {\textit{Remark.}} As previously discussed, Lemma~\ref{lem:no_phi} does not apply to 
  consensus, since we know that there exists a formula proving its unsolvability.
  The reason is that 
	  the projection mapping $\pi : \O \to \I$ in consensus does not induce a
	  bisimulation.
  Here we show that condition~\ref{forth} of Definition~\ref{def:bisimulation} does not hold.
  Namely, if $X = \{(\B,0), (\W,1)\}$ and $X' = \{(\B,0,1), (\W,1,1)\}$ and 
      $Y = \{(\B,0), (\W,0)\}$, then by definition of consensus there cannot exist 
	  a facet~$Y'$ with $Y \mathrel{R} Y'$ and $\B\in\chi'(X' \cap Y')$.
  Such a facet would have the form $Y'=\{(\B,0,1),(\W,0,d)\}$, for a $d\in\{0,1\}$,
	  which is not a valid world in the output model of consensus for any decision~$d$.
}

\subsection{Extended DEL}
\label{sec:extendedDEL}

In Section~\ref{sec:bisimulation_and_limits}, we have shown that no epistemic
logic formula is able to express the reason why the equality negation task is not solvable.
This seems to indicate that our logic is too weak: indeed, because of the
product update model construction that we use, we are only allowed
to write formulas about the inputs and what the processes know about each
other's inputs.
But the specification of the task is very much about the outputs too!
If we allow ourselves to use atomic propositions of the form
$\decideprop{p}{d}$, with the intended meaning that process~$p$ decides
value $d$, a good candidate for the formula $\phi$ seems to be:
\[ \phi \; = \; \bigwedge_{p,i,d} \inputprop{p}{i} \land \decideprop{p}{d}
  \implies \left( (\inputprop{\bar p}{i} \land \decideprop{\bar p}{\bar d})
  \lor (\inputprop{\bar p}{\bar i} \land \decideprop{\bar p}{d})
  \right)
\]
where $\bar p$, $\bar i$, $\bar d$ denote values different from $p$, $i$, $d$,
respectively. Note that $\bar p$ and $\bar d$ are uniquely defined (since there
are only two processes and two decision values), but for $\bar i$, there are
two possible inputs different from $i$. So, for example,
$\inputprop{p}{\bar 0}$ is actually a shortcut for $\inputprop{p}{1} \lor
\inputprop{p}{2}$.

This formula simply expresses the specification of the task: if process $p$ has
input $i$ and decides $d$, then the other process should either have the same
input and decide differently, or have a different input and decide the same.
Then hopefully $\phi$ would be true in every world of the output complex,
but would fail somewhere in the protocol complex $\I[\IS]$, meaning that
the $N$-layer message-passing model is not powerful enough to obtain this knowledge.

\medskip

To be able to express such a formula, we first need to enrich our models by
saying in which worlds the atomic propositions $\decideprop{p}{d}$ are true
or false.
Let $\ExtAtom = \Atom \cup \{\decideprop{p}{d} \mid p \in \Agents, d \in \{0,1\}\}$
be the new set of atomic propositions.
The definition of the \emph{extended product update model} $\ExtProd{\I}{\T} = \ExtO$ is straightforward:
\begin{itemize}
\item Its vertices are of the form $(p,i,d)$ with $p \in \Agents$, $i \in
\{0,1,2\}$ and $d \in \{0,1\}$. The facets are $\{(\B,i,d_\B),(\W,j,d_\W)\}$
where $i=j \iff d_\B \neq d_\W$.
\item The coloring map is $\chi_{\extended{\O}}(p,i,d) = p$.
\item The atomic propositions labeling is $\ell_{\extended{\O}}(p,i,d) = \{\inputprop{p}{i},
  \decideprop{p}{d}\}$.
\end{itemize} 
Thus, this is almost the same model as the one of
Section~\ref{sec:DEL_analysis}, but we have added some annotations to say where
the $\decideprop{p}{d}$ atomic propositions are true.
It is easily checked that the formula $\phi$ is true in every world of $\ExtO$.

\medskip

Now, we would also like the formula $\phi$ to make sense in the protocol
complex $\I[\IS]$, but it does not seem to have any information about
decision values. It only describes the input values, and which execution has
occurred. But it is precisely the role of the simplicial map $\delta : \I[\IS]
\to \O$ to assign decision values to each world of $\I[\IS]$.
Thus, given such a map $\delta$, we can lift it to a map $\extended{\delta} :
\ExtProd{\I}{\IS} \to \ExtO$ as the following lemma states.

\begin{lemma}
\label{lem:lifting_model}
Let $\M = \langle M, \chi, \ell \rangle$ be a simplicial model over the set of
agents $\Agents$ and atomic propositions $\Atom$,
and let $\delta : \M \to \O$ be a morphism of simplicial models.
Then there is a unique model $\extended{\M} = \langle M, \chi, \extended{\ell}
\rangle$ over $\ExtAtom$, where $\extended{\ell}$ agrees with $\ell$
on $\Atom$, such that $\extended{\delta} : \extended{\M} \to \ExtO$ is still a morphism
of simplicial models.
\end{lemma}
\begin{proof}
All we have to do is label the worlds of $\M$ with the $\decideprop{p}{d}$
atomic propositions, such that $\delta$ is a morphism of simplicial models.
Thus, we define $\extended{\ell} : M \to \powerset{\ExtAtom}$ as
$\extended{\ell}(m) = \ell(m) \cup \{\decideprop{p}{d}\}$, where $\delta(m) =
(p,i,d) \in O$.
Then~$\delta$ is still a chromatic simplicial map (since we did not change the
underlying complexes nor their colors), and moreover we have
$\extended{\ell}(m) = \ell_{\extended{\O}}(\delta(m))$ for all $m$.
The model $\extended{\M}$ is unique since any other choice of
$\extended{\ell}(m)$ would have broken this last condition, so $\delta$ would
not be a morphism of simplicial models. \qed
\end{proof}

We can finally prove that the equality negation task is not solvable:

\begin{theorem}
The equality negation task for two processes is not solvable in the $N$-layer message-passing model.
\end{theorem}
\textit{Proof.}\;
Let us assume by contradiction that the task is solvable, i.e., by Definition~\ref{def:solving_task},
there exists a morphism of simplicial models $\delta : \I[\IS] \to \O$ that
makes
\begin{wrapfigure}[9]{r}{0.35\textwidth}%
\centering
 \vspace{-16pt}
\begin{tikzpicture}
  \node (s) {$\cI[\IS]$};
  \node (xy) [below=2 of s] {$\O$};
  \node (s') [right=0.3 of s] {$\ExtProd{\cI}{\IS}$};
  \node (xy') [below=1.85 of s'] {$\ExtO$};
  \node (x) [left=of xy] {$\cI$};
  \draw[<-] (x) to node [sloped, above] {$\pi$} (s);
  \draw[->, right] (s) to node {$\delta$} (xy);
  \draw[->, right] (s') to node {$\widehat{\delta}$} (xy');
  \draw[->] (xy) to node [below] {$\pi$} (x);
\end{tikzpicture}
\end{wrapfigure}

\noindent
the diagram commute.
By Lemma~\ref{lem:lifting_model}, we can lift~$\delta$ to a morphism
$\extended\delta : \ExtProd{\I}{\IS} \to \ExtO$ between the extended models.
As we remarked earlier, the formula $\phi$ is true in every world of
$\ExtO$. Therefore, it also has to be true in every world of $\ExtProd{\I}{\IS}$.
Indeed, for any world $w$, since $\ExtO,\delta(w) \models \phi$, and~$\delta$ is a morphism, by Lemma~\ref{lem:gain_knowledge}, we must have
$\ExtProd{\I}{\IS},w \models \phi$.
We will now derive a contradiction from this fact.

Recall that the protocol complex $\I[\IS]$ is just a subdivision of the input
complex $\I$, as depicted below. (For simplicity, some input values have been omitted in the
vertices on a subdivided edge; it is the same input as the extremity of the
edge which has the same color. Also, the picture shows only one subdivision, 
but our reasoning is unrestricted and it applies to any number of layers~$N$.)
\begin{center}
\begin{tikzpicture}[auto, scale=2, whitevertex/.style={draw=black,thick,circle,fill=white,inner sep=2pt,minimum size=6pt}, blackvertex/.style={draw=black,thick,circle,fill=black,inner sep=2pt,minimum size=6pt,text=white}]

\node[whitevertex] (w0) at (0,0) {$0$};
\node[whitevertex] (w1) at (1,1) {$1$};
\node[whitevertex] (w2) at (2,0) {$2$};
\node[blackvertex] (b0) at (0,1) {$0$};
\node[blackvertex] (b1) at (1,0) {$1$};
\node[blackvertex] (b2) at (2,1) {$2$};
\path (w0) edge (b0) edge
			node [below, pos=0.12] {$w_1$}
			node [below, pos=0.55] {$w_2$} (b1)
		edge (b2)
	  (w1) edge (b0) edge (b1) edge (b2)
      (w2) edge (b0) edge (b1) edge
      		node[right, pos=0.9] {$w'$}
      		node[right, pos=0.1] {$w''$} (b2);
\foreach \i in {0,1,2} {
  \foreach \j in {0,1,2} {
    \node[blackvertex] at ($(w\i)!0.33!(b\j)$) {};
    \node[whitevertex] at ($(b\i)!0.33!(w\j)$) {};
  }
}
\end{tikzpicture}
\end{center}

Let us start in some world $w_1$ on the $(\W,0)-(\B,1)$ edge.
In the world $w_1$, the two processes have different inputs. Since in
$\ExtProd{\I}{\IS}$, the formula $\phi$ is true in $w_1$, the decision values
have to be the same. Without loss of generality, let us assume that in $w_1$,
both processes decide $0$.

We then look at the next world $w_2$, which shares a black vertex with $w_1$.
Since the inputs are still $0$ and $1$, and $\phi$ is true, and we assumed
that process $\B$ decides $0$, then the white vertex of $w_2$ also has to
decide $0$.

We iterate this reasoning along the $(\W,0)-(\B,1)$ edge, then along the
$(\B,1)-(\W,2)$ edge, and along the $(\W,0)-(\B,2)$ edge: all the vertices on
these edges must have the same decision value $0$. 
Thus, on the picture, the top right $(\B,2)$ corner has to decide~$0$, as well
as the bottom right $(\W,2)$ corner.

Now in the world $w'$, the two input values are equal,
so the processes should decide differently.
Since the black vertex decides $0$, the white vertex must have decision
value~$1$.
If we keep going along the rightmost edge, the decision values must alternate:
all the black vertices must decide~$0$, and the white ones decide~$1$.
Finally, we reach the world $w''$, where both decision values are~$0$, whereas
the inputs are both~$2$. So the formula $\phi$ is false in $w''$, which
is a contradiction. 

\qed
\medskip

It is interesting to compare the epistemic formula $\phi$ that we used in this paper to prove the unsolvability of equality negation, with the one (let us call it $\psi$) that was used in~\cite{gandalf2018} to prove the impossibility of solving consensus.
In the case of consensus, we did not need the ``Extended DEL'' framework. The formula $\psi$ was simply saying that the processes have common knowledge of the input values. This formula is quite informative: it tells us that the main goal of the consensus task is to achieve common knowledge.
On the other hand, the formula $\phi$ is less informative: it is simply stating the specification of the equality negation task.
\longversion{
  It does not even seem to be talking about knowledge, since there
  are no $K$ or $C$ operators in the formula.
  In fact, the epistemic content of $\phi$ is hidden in the $\decideprop{p}{d}$
  atomic propositions. Indeed, their semantics in $\ExtProd{\cI}{\IS}$ is
  referring to the decision map $\delta$, which assigns a decision value $d$ to
  each vertex of $\cI[\IS]$. The fact that we assign decisions to vertices means
  that each process must decide its output \emph{solely according to its
  knowledge}.
}
\medskip

Despite the fact that it produces less informative formulas, the ``Extended DEL'' proof method seems to be able to prove \emph{any} impossibility result.
Let $\cT = \langle T, \sim, \pre \rangle$ be a task action model, on the input model $\cI$, and let $\cP$ be a protocol action model.
Remember that the elements of $T$ are functions $t : \Agents \to \Vout$ assigning a decision value to each agent.
Let $\phi$ denote the following formula:
\begin{equation}\label{formula}
\phi \; = \; \bigwedge_{X \in \cF(\cI)} \left( \bigwedge_{p \in \Agents}
\inputprop{p}{X(p)} \implies \bigvee_{\substack{t \in T\\ \cI,X \models\, \pre(t)}} \;
\bigwedge_{p \in \Agents} \decideprop{p}{t(p)} \right)
\end{equation} 
where $X(p)$ denotes the input value of process $p$ in the input simplex $X$.
Then we get the following Theorem (whose proof is in Appendix~\ref{app:specification_proof}).

\begin{theorem}
\label{thm:logical_specification}
The task $\cT$ is solvable in the protocol $\cP$ if and only if there exists an extension $\ExtProd{\cI}{\cP}$ of $\cI[\cP]$ (assigning a single decision value to each vertex of $\cI[\cP]$) such that $\phi$ from (\ref{formula}) is true in every world of $\ExtProd{\cI}{\cP}$. 
\end{theorem}

\begin{proof}
($\Rightarrow$): We already proved this direction in Section~\ref{sec:extendedDEL}. Assume that the task is solvable, i.e., by Definition~\ref{def:solving_task}, there is a morphism $\delta : \cI[\cP] \to \cI[\cT]$ such that $\pi \circ \delta = \pi$.
The model $\ExtProd{\cI}{\cP}$ is given Lemma~\ref{lem:lifting_model}, where the assignment of decision values is the one given by $\delta$. Then $\phi$ is easily seen to be true in every world of $\ExtProd{\cI}{\cT}$, and by Lemma~\ref{lem:gain_knowledge}, it is also true in every world of $\ExtProd{\cI}{\cP}$.
\medskip

\noindent
($\Leftarrow$): For the converse, assume that there is a model $\ExtProd{\cI}{\cP}$ where the formula~$\phi$ is true in every world.
Then we build a map $\delta : \cI[\cP] \to \cI[\cT]$ as follows.
If a vertex $(i,E)$ of $\ExtProd{\cI}{\cP}$, colored by agent $p$, is labeled with the atomic proposition $\decideprop{p}{d}$, we send it to the vertex $\delta(i,E) := (i,d)$ of $\cI[\cT]$.

By definition, we have the commutative diagram $\pi(i,E) = i = \pi \circ \delta(i,E)$. We now need to show that~$\delta$ is a morphism of simplicial models.
The coloring and labeling maps are preserved, since by definition they just copy the coloring and labeling of $\cI$. We still have to prove that $\delta$ sends simplices to simplices.

Let $Y$ be a facet of $\cI[\cP]$. Let $X = \pi(Y) \in \cF(\cI)$ be the facet of the input complex corresponding to the initial values of the processes in the execution~$Y$.
Then we have $\cI[\cP], Y \models \bigwedge\limits_{p \in \Agents}
\inputprop{p}{X(p)}$, and since $\ExtProd{\cI}{\cP}, Y \models \phi$, by modus ponens there must be some action $t \in T$, with $\cI, X \models \pre(t)$, such that $\ExtProd{\cI}{\cP}, Y \models \bigwedge\limits_{p \in \Agents} \decideprop{p}{t(p)}$. 
Thus, the vertex $v$ of $Y$ which is colored by $p$ must be labeled with the atomic proposition $\decideprop{p}{t(p)}$, and so the map $\delta$ sends it to the vertex $(i,t(p))$ of $\cI[\cT]$.
By definition of the action model, since $\cI, X \models \pre(t)$, the set of vertices $\delta(Y) = \{ (i, t(p)) \mid i \in X, p = \chi(i) \}$ is a simplex of $\cI[\cT]$.
Therefore, the map $\delta$ is a simplicial map.

\qed
\end{proof}

This theorem implies that the situation of Section~\ref{sec:bisimulation_and_limits} cannot happen with the ``Extended DEL'' approach: if the task is not solvable, there necessarily exists a world $X$ of $\ExtProd{\cI}{\cP}$ where the formula fails. Of course, finding such a world is usually the hard part of an impossibility proof, but at least we know it exists.
In fact, in the particular case of read/write protocols (or, equivalently, layered message-passing), the solvability of tasks is known to be undecidable when there are more than three processes~\cite{GafniK99,HerlihyR97}. Thus, according to our Theorem, given a formula $\phi$, the problem of deciding whether there exists a number of layers~$N$ and an extension of $\cI[\IS]$ which validates the formula, is also undecidable.

Notice that in Theorem~\ref{thm:logical_specification}, we characterized the solvability of a task without referring to~$\cT$ itself: the formula $\phi$ contains all the information of $\cT$. Thus,  instead of relying on the commutative diagram of Definition~\ref{def:solving_task}, one can specify a task directly as a logical formula.

\longversion{
  \section{Conclusion}

The equality negation task is known to be unsolvable in the wait-free read/write model.
In this paper, we gave a new proof of this result, using the simplicial complex semantics of DEL that we proposed in~\cite{gandalf2018}.
There are two purposes of doing this.
First, the logical formula witnessing the unsolvability of a task usually helps us understand the epistemic content of this task. Unfortunately, as it turns out, the logical formula that we obtained in the end is less informative than we hoped.
Secondly, this is a nice case study to test the limits of our DEL framework.
Indeed, we proved in Section~\ref{sec:bisimulation_and_limits} that the basic language of DEL, where formulas are only allowed to talk about input values, is too weak to express the reason why the task is not solvable.
To fix this issue, we introduced a way to extend our logical language in order to have more expressive formulas.

}

%
%
%
%
%
%
\section*{Acknowledgements}

The authors were supported by 
the academic chair ``Complex Systems Engineering'' of Ecole Polytechnique-ENSTA-T\'{e}l\'{e}com-Thal\`es-Dassault-DGA-Naval Group-FX-FDO-Fondation ParisTech, 
by DGA project ``Validation of Autono\-mous Drones and Swarms of Drones'', 
by the UNAM-PAPIIT project IN109917, 
by the France-Mexico Binational SEP-CONACYT-ANUIES-ECOS grant M12\-M01,
by the European Research Council (ERC) 
under the European Union's Horizon 2020 research and innovation 
programme under grant agreement No 787367 (PaVeS), as well as by 
the Austrian Science~Fund~(FWF) through Doctoral College LogiCS (W1255-N23).


\newpage

\appendix

\section{Distributed computing through combinatorial topology}

In this appendix, we briefly describe the usual topological approach of task computability, focusing on the special case of two processes; additional details appear for example in~\cite{HerlihyKR:2013}.
The DEL framework that we expose in Section~\ref{sec:prelim} is based on this approach. For instance, the notion of \emph{carrier map} in the context of the topological approach, is replaced by the product update operation in DEL.

\medskip

\subsection{The topological definition of task solvability}

To model a computation for $n+1$ processes,  a simplicial complex of dimension $n$ is used; in the case of two processes, a $1$-dimensional simplicial complex is simply a graph.
Thus, to simplify the formalism, we  use here graph-theoretic notions instead
of simplicial complexes.


For an (undirected) graph $\cG$, we write $\cV(\cG)$ for the set of vertices of $\cG$. To mimic the usual vocabulary of the simplicial approach to distributed computability, we use the word \emph{simplex} to refer either to a vertex or an edge of the graph. The edges will usually be identified with two-element sets $\{s, s'\}$, where $s,s' \in \cV(\cG)$ are the two endpoints of the edge; and vertices are identified with singletons $\{s\}$ for $s\in \cV(\cG)$. 
 We say that a graph
is \emph{chromatic} if it is
equipped with a coloring $\chi:\cV(\cG)\rightarrow C$,
where $C$ is a finite set of colors, such that
the vertices of every edge have distinct colors.
A graph $\cG$ is \emph{labeled} if it is equipped with a labeling $\ell:\cV(\cG) \to L$, where $L$ is a set of labels (without any extra condition).
To model a computation between two processes $B$ and $W$, we use a chromatic labeled graph $\langle \cG, \chi, \ell \rangle$.
The set of colors $C = \set{B,W}$ is the set of process names. 
The labeling $\ell(s)$ of $s$ is called the \emph{view} of the corresponding process.
Moreover, we require that each vertex in a chromatic
labeled graph is uniquely identified by its name $\chi(s)$ and label $\ell(s)$.

Given two graphs $\cG$ and $\cH$, a map $\mu: \cV(\cG) \to \cV(\cH)$ from the vertices of~$\cG$ to the vertices of~$\cH$ is a
 \emph{simplicial map} if whenever
$\set{\vs_0,\vs_1}$ is an edge in~$\cG$, then
$\set{\mu(\vs_0),\mu(\vs_1)}$ is either an edge or a vertex of $\cH$.
If $\vs_0 \neq \vs_1$ implies
$\mu(\vs_0) \neq \mu(\vs_1)$, then the map is said to be \emph{rigid}.
(Thus, a rigid simplicial map is  a~{graph homomorphism}).
When $\cG$ and $\cH$ are chromatic, we usually assume that the
simplicial map $\mu$ preserves names: $\chi(\vs)=\chi(\mu(\vs))$.
Thus, {chromatic simplicial} maps are rigid.
The following is an important property of simplicial maps, easy to check.
\begin{lemma}
\label{lem:connectSM}
The image of a connected graph  under a simplicial map is
connected.
\end{lemma}

Given two graphs $\cG$ and $\cH$,
a~\emph{carrier map} $\Phi$ from $\cG$ to $\cH$, written $\Phi:\cG \to 2^\cH$, takes each
simplex $\sigma\in\cG$ to a~subgraph $\Phi(\sigma)$ of~$\cH$,
such that
$\Phi$ satisfies the following \emph{monotonicity} property:
for all simplices $\sigma,\tau\in\cG$,
if $\sigma\subseteq\tau$,
then $\Phi(\sigma)\subseteq\Phi(\tau)$.
 If the set of colors of the subgraph  $\Phi(\sigma)$
is equal to the set of colors of $\sigma$, then $\Phi$ is \emph{name-preserving}.

Let $\Vin$ be a set of \emph{input values}, and $\Vout$ a set of \emph{output
values}. A \emph{task} for $B$ and $W$ is a triple
$(\cI, \cO, \Delta)$, where
\begin{itemize}
\item $\cI$ is a pure chromatic \emph{input graph}
colored by $\set{B,W}$ and labeled by $\Vin$; 
\item $\cO$ is a pure chromatic \emph{output graph}
colored by $\set{B,W}$ and labeled by $\Vout$;
\item $\Delta$ is a name-preserving carrier map from $\cI$ to
$\cO$.
\end{itemize}
The input graph defines all the possible ways the two processes can
start the computation, the output graph defines all the possible
ways they can end, and the carrier map defines which input can lead to
which outputs.
More precisely, each edge $\set{(B,i), (W,j)}$ in $\cI$ defines a possible input configuration
(initial system state), where the local state of $B$ consists of the
 input value $i \in \Vin$ and the local state of $W$ consists of  input value $j \in \Vin$.
The processes communicate with one another,
and each eventually decides on an output value and halts.
If~$B$ decides~$x$, and $W$ decides $y$, then there is an output configuration
represented by an edge $\set{(B,x), (W,y)}$ in the output graph.
If the following condition is verified,
\begin{equation*}
\set{(B,x), (W,y)} \in \Delta(\set{(B,i), (W,j)}),
\end{equation*}
then we say that this run respects the task specification $\Delta$.

We now turn our attention from tasks,
to  the model of computation in which we want to solve them.
Consider a protocol execution in which the processes exchange information
through the channels (message-passing, read-write memory, or other) provided by the
model.
At the end of the execution,
each process has its own view (final state).
The set of all possible final views themselves form a chromatic graph.
Each vertex is a pair $(P,p)$, where $P$ is a process name,
and~$p$ is the view (final state) of $P$ at the end of some execution.
A pair of such vertices $\set{(B,p),(W,q)}$ is an edge if there is some
execution where $B$ halts with view $p$ and $W$ halts with view $q$.
This graph is called the \emph{protocol graph}.

There is a carrier map $\Xi$ from $\cI$ to $\cP$,
called the \emph{execution carrier map},
that carries each input simplex to a subgraph of the protocol graph.
$\Xi$ carries each input vertex $(P,v)$ to the solo execution in which $P$
finishes the protocol without hearing from the other process.
It carries each input edge $\set{(B,i),(W,j)}$ to the subgraph of
executions where $B$ starts with input $i$ and $W$ with input $j$.

\index{decision map}
The protocol graph is related to the output graph by a \emph{decision map}
$\delta$ that sends each protocol graph vertex $(P,p)$ to an output graph
vertex $(P,w)$, labeled with the same name.
Operationally, this map should be understood as follows:
if there is a protocol execution in which $P$ finishes with view $p$
and in this view it chooses output $w$,
then $(P,p)$ is a vertex in the protocol graph,
$(P,w)$ a vertex in the output graph,
and $\delta((P,p)) = (P,w)$.
It is easy to see that $\delta$ is a simplicial map,
carrying edges to edges,
because any pair of mutually compatible final views
yields a pair of mutually compatible decision values.

The composition of the decision map $\delta:\cP \to \cO$ with the carrier
map $\Xi:\cI \to 2^\cO$ is a carrier map $\Phi = \delta \circ \Xi:\cI~\rightarrow~2^\cO$, 
and we say that $\Phi$ is \emph{carried by} $\Delta$
whenever $\Phi(\sigma) \subseteq \Delta(\sigma)$ for every simplex $\sigma\in\cI$.
%
We can now define what it means for a protocol to solve a task.
\begin{definition}
\label{def:solving}
The protocol $(\cI,\cP,\Xi)$ \emph{solves} the task $(\cI, \cO, \Delta)$ 
if there exists a simplicial \emph{decision map} $\delta$ from $\cP$ to $\cO$
such that $\delta \circ \Xi$ is carried by $\Delta$.
\end{definition}

\subsection{The layered message-passing model}
\label{app:layered_model}

We now define the $N$-layer protocol of Section~\ref{sec:layered_model} using the carrier map formalism described above.
Recall that, given an initial state, an execution can be specified by a sequence of $N$ symbols over the alphabet $\set{\bot,B,W}$,
meaning that, in the $i$-th layer,  if the $i$-th symbol is $\bot$ then both messages arrived,
and if it is $B$  (resp. $W$) then only $B$'s message failed to arrive (resp. $W$).

Let us first take a look at the $1$-layer protocol graph.
Consider first that the input graph $\cI$ consists of just one edge $\sigma\in\cI$, 
$\sigma=\set{(B,i),(W,j)}$, where $i$ and $j$ are input values, depicted below.
\begin{center}
\begin{tikzpicture}[auto, scale=2, whitevertex/.style={draw=black,thick,circle,fill=white,inner sep=2pt,minimum size=1em}, blackvertex/.style={draw=black,thick,circle,fill=black,inner sep=2pt,minimum size=1em,text=white}]
\node[blackvertex] (b0) at (0,1) {$i$};
\node[whitevertex] (w1) at (1,1) {$j$};
\path[thick]	  (b0) edge  node[above]  {$\sigma$}   (w1);
\end{tikzpicture}
\end{center}
In the single-layer protocol graph, $\Xi(\sigma)$ is a path of three edges,
corresponding resp. to the three executions $W$, $\bot$, and $B$:
\begin{center}
\begin{tikzpicture}[auto, scale=2, whitevertex/.style={draw=black,thick,circle,fill=white,inner sep=2pt,minimum size=1em}, blackvertex/.style={draw=black,thick,circle,fill=black,inner sep=2pt,minimum size=1em,text=white}]
\node[blackvertex] (b0) at (0,1) {$i\emptyview$};
\node[whitevertex] (w1) at (1,1) {$ij$};
\node[blackvertex] (b2) at (2,1) {$ij$};
\node[whitevertex] (w3) at (3,1) {$\emptyview j$};
\path[thick]	  (b0) edge  node[above]  {$W$}   (w1)
          (w1) edge node[above]  {$\bot$} (b2)
          (b2) edge node[above]  {$B$} (w3);
\end{tikzpicture}
\end{center}
Formally, $\Xi(\sigma)$ is the following simplicial complex:
\begin{equation*}
  \left\{ 
  \set{(B,i\emptyview), (W,ij)}, \set{(B,ij), (W,ij)}, \set{(B,ij), (W, \emptyview j)}
  \right\},
\end{equation*}
where $(X,yz)$ denotes a vertex colored with process name $X$, which received
message $y$ from $B$, and message $z$ from $W$.

The symbol $\emptyview$ in the view of a process indicates that no message was received by that process. 
For example, starting with inputs $i$ and $j$, if the execution~`$B$' happens (meaning that the message sent by process $B$ was lost), then the white process receives no message from $B$ and has the view $\emptyview j$, whereas the black process receives the message from $W$ and has the view $ij$.
Notice that the white vertex $\emptyview j$ belongs to only one edge (the white process knows that the execution was `$B$'), whereas the black vertex $ij$ belows to two edges (from the point of view of the black process, the execution could be either `$B$' or `$\bot$').
\medskip

\noindent
Now suppose that the input graph $\cI$ has a second input edge $\tau = \{(B,k), (W,j)\}$.
\begin{center}
\begin{tikzpicture}[auto, scale=2, whitevertex/.style={draw=black,thick,circle,fill=white,inner sep=2pt,minimum size=1em}, blackvertex/.style={draw=black,thick,circle,fill=black,inner sep=2pt,minimum size=1em,text=white}]
\node[blackvertex] (b0) at (0,1) {$i$};
\node[whitevertex] (w1) at (1,1) {$j$};
\node[blackvertex] (b2) at (2,1) {$k$};
\path[thick]	  (b0) edge  node[above]  {$\sigma$}   (w1)
          (w1) edge  node[above]  {$\tau$}   (b2);
\end{tikzpicture}
\end{center}
Then the protocol graph for the single-layer protocol is as follows:
\begin{center}
\scalebox{0.8}{
\begin{tikzpicture}[auto, scale=2, whitevertex/.style={draw=black,thick,circle,fill=white,inner sep=2pt,minimum size=1em}, blackvertex/.style={draw=black,thick,circle,fill=black,inner sep=2pt,minimum size=1em,text=white}]
\node[blackvertex] (b0) at (0,1) {$i\emptyview$};
\node[whitevertex] (w1) at (1,1) {$ij$};
\node[blackvertex] (b2) at (2,1) {$ij$};
\node[whitevertex] (w3) at (3,1) {$\emptyview j$};
\node[blackvertex] (b4) at (4,1) {$kj$};
\node[whitevertex] (w5) at (5,1) {$kj$};
\node[blackvertex] (b6) at (6,1) {$k\emptyview$};

\path[thick]	  (b0) edge  node[above]  {$W$}   (w1)
          (w1) edge node[above]  {$\bot$} (b2)
          (b2) edge node[above]  {$B$} (w3)
          (w3) edge  node[above]  {$B$}   (b4)
          (b4) edge node[above]  {$\bot$} (w5)
          (w5) edge node[above]  {$W$} (b6);
          
\draw [decorate,decoration={brace,amplitude=10pt,mirror,raise=12pt}]
(b0.west) -- (w3.east) node [black,midway,below, yshift=-22pt] {$\Xi(\sigma)$};
\draw [decorate,decoration={brace,amplitude=10pt,mirror,raise=12pt}]
(b6.east) -- (w3.west) node [black,midway,above, yshift=22pt] {$\Xi(\tau)$};

\end{tikzpicture}
}
\end{center}

Notice how the white vertex $\emptyview j$ in the middle now belongs to two edges. Indeed, even if it knows for sure that the execution was `$B$', it did not receive the message from $B$, so it does not know whether the input value of black was~$i$ or $k$. Hence, this vertex belongs to both $\Xi(\sigma)$ and $\Xi(\tau)$.
It is remarkable that the single-layer protocol graph in this model is the same
as the input graph, except that each input edge is subdivided into three.
\medskip

In the $N$-layer message passing protocol, we iterate the previous construction $N$ times.
At each subsequent layer, each process uses its view from the previous layer as its input value for the next layer.
For example, starting with input values~$i$ and $j$, as we saw previously, the execution `$B$' in the first layer gives the two views $\set{(B,ij), (W, \emptyview j)}$.
Then the inputs for the the second layer are $i' := ij$ and $j' := \emptyview j$, and 
the $2$-layer execution `$BB$' gives the views $\set{(B,i'j'), (W, \emptyview j')} = \set{(B,ij\;\; \emptyview j), (W, \emptyview\;\; \emptyview j)}$.

It should now be clear that each subsequent layer further subdivides the edges of the previous
layer into three.  It carries each input edge $\set{(B,i),(W,j)}$ to the subgraph of
executions where $B$ starts with input $i$ and $W$ with input $j$.
Formally, 

\begin{theorem}
\label{th:mainSubdiv}
The protocol graph $\cP$
for the $N$-layer message passing model is a subdivision of the input graph, where each edge is divided into~$3^N$ edges.
The  execution carrier map $\Xi$ from $\cI$ to $\cP$,
 carries each input edge to its subdivision in the protocol graph, and
 each input vertex $(P,v)$ to the solo execution in which~$P$
finishes the protocol without hearing from the other process.
\end{theorem}

For completeness, we define formally by induction on~$N$ the \emph{view} of a process in a $N$-layer execution, starting from the input edge $\sigma = \{(B,i), (W,j)\}$.
For $N=0$, we have $\view_B(\varnothing, \sigma) = i$ and ${\view_W(\varnothing, \sigma) = j}$.
Let $\alpha x$ be a $N$-layer execution, where $\alpha$ is a $(N-1)$-layer execution and $x \in \{B, W, \bot\}$. Assume that, after the execution $\alpha$ occurred, the view of $B$ is $V_B := \view_B(\alpha, \sigma)$, and the view of $W$ is $V_W := \view_W(\alpha, \sigma)$.
Then:
\begin{itemize}
\item if $x = B$, then $\view_B(\alpha x, \sigma) = (V_B, V_W)$ and $\view_W(\alpha x, \sigma) = (\emptyview, V_W)$.
\item if $x = W$, then $\view_B(\alpha x, \sigma) = (V_B, \emptyview)$ and $\view_W(\alpha x, \sigma) = (V_B, V_W)$.
\item if $x = \bot$, then $\view_B(\alpha x, \sigma) = (V_B, V_W)$ and $\view_W(\alpha x, \sigma) = (V_B, V_W)$.
\end{itemize}

\noindent
Finally, the protocol graph $\cP_N$ of the $N$-layer message passing protocol with input graph $\cI$ has vertices of the form $(X, \view_X(\alpha, \sigma))$ with $X \in \{B, W\}$, $\alpha \in \{B, W, \bot\}^N$ and $\sigma \in \cI$.
Its edges are of the form $\{(B, \view_B(\alpha, \sigma)), (W, \view_W(\alpha, \sigma)) \}$, with $\alpha \in \{B, W, \bot\}^N$ and $\sigma \in \cI$.
Thus, a vertex of color $X$ belongs to two such edges whenever $\view_X(\alpha, \sigma) = \view_X(\beta, \tau)$, for some $\alpha, \beta, \sigma, \tau$.
The carrier map $\Xi : \cI \to 2^\cP$ takes each edge $\sigma \in \cF(\cI)$ to the subgraph of $\cP$ containing all the vertices of the form $(X, \view_X(\alpha, \sigma))$, and all edges between them.
\medskip

Two different choices for the pair $(\alpha, \sigma)$ will necessarily give rise to two distinct edges in $\cP$, as stated in the following Lemma:

\begin{lemma}
\label{lem:equal_views}
For all $N, \alpha, \beta, \sigma, \tau$, if $\view_B(\alpha, \sigma) = \view_B(\beta, \tau)$ and $\view_W(\alpha, \sigma) = \view_W(\beta, \tau)$, then $\alpha = \beta$ and $\sigma = \tau$.
\end{lemma}
\begin{proof}
By induction on $N$.
\qed
\end{proof}

\begin{lemma}
\label{lem:model_isomorphic}
The protocol graph $\cP_N$ for the $N$-layer message-passing protocol is isomorphic to the product update model $\cI[\IS]$ defined in Section~\ref{sec:layered_model}.
\end{lemma}
\begin{proof}
The main property that we need to prove is the following, for any number of layers~$N$, for all $\alpha, \beta \in L_N$ and $\sigma, \tau \in \cF(\cI)$:
\begin{equation}
\view_B(\alpha, \sigma) = \view_B(\beta, \tau) \; \iff \; (\alpha, \sigma) \sim_B (\beta, \tau)
\end{equation}
and similarly for $\view_W$ and $\sim_W$.
We prove it by induction on $N$.
\begin{itemize}
\item For $N=0$, this holds from the definitions.
\item Let $\alpha x$ and $\beta y$ be executions of length $N+1$, with $x, y \in \{B, W, \bot\}$.
\medskip

($\Leftarrow$): Assume $(\alpha x, \sigma) \sim_B (\beta y, \tau)$. By definition, there are two possible cases. 
  \begin{itemize}
  \item $x = y = W$ and $(\alpha, \sigma) \sim_B (\beta, \tau)$. By induction hypothesis, we get $\view_B(\alpha, \sigma) = \view_B(\beta, \tau)$. Let us write $V_B$ for this view.
  Then, by definition, we get $\view_B(\alpha x, \sigma) = (V_B, \emptyview)$ and $\view_B(\beta y, \tau) = (V_B, \emptyview)$, which are equal.
  \item $x, y \in \{\bot, B\}$ and $\sigma = \tau$ and $\alpha = \beta$.
  We write $V_B = \view_B(\alpha, \sigma) = \view_B(\beta, \tau)$, and $V_W = \view_W(\alpha, \sigma) = \view_W(\beta, \tau)$.
  Then by definition $\view_B(\alpha x, \sigma) = (V_B, V_W) = \view_B(\beta y, \tau)$.
  \end{itemize}

($\Rightarrow$): Assume $\view_B(\alpha x, \sigma) = \view_B(\beta y, \tau)$.
Since a view can never be just `$\emptyview$' (it is either an input value, or a pair of values), we can never have $V_W = \emptyview$. Thus the equality can be true only in one of the two following cases.
  \begin{itemize}
  \item Either $x = y = W$. In which case, we get $\view_B(\alpha, \sigma) = \view_B(\beta, \tau)$ by identifying the first components of the views, and by induction hypothesis, $(\alpha, \sigma) \sim_B (\beta, \tau)$. Therefore, $(\alpha x, \sigma) \sim_B (\beta y, \tau)$.
  \item Or $x, y \in \{\bot, B\}$. By identifying the components of the pair, we get $\view_B(\alpha, \sigma) = \view_B(\beta, \tau)$ and $\view_W(\alpha, \sigma) = \view_W(\beta, \tau)$. By Lemma~\ref{lem:equal_views} we deduce $\alpha = \beta$ and $\sigma = \tau$, and thus we obtain $(\alpha x, \sigma) \sim_B (\beta y, \tau)$.
  \end{itemize}
\end{itemize}
The proof for the correspondence between $\view_W$ and $\sim_W$ is similar, with the role of $B$ and $W$ reversed.

Once we have shown (1), proving the Lemma is just a matter of unfolding the definitions of the product update model.
A vertex of $\cI[\IS]$ is formally given by a pair $(v, E)$, where $v$ is a vertex of $\cI$ (say, of color $B$) and $E$ is an equivalence class of $\sim_B$.
Let $(\alpha, \sigma) \in E$ be an action in $E$. Then, to the vertex $(v,E)$ of $\cI[\IS]$ we associate the vertex $(B, \view_B(\alpha, \sigma))$ of $\cP_N$; this is well-defined thanks to (1). Surjectivity is obvious.
To show injectivity, assume that $(v,E)$ and $(v',E')$ give the same view. By (1), we get $E=E'$; and $v=v'$ according to the precondition $\pre$ of the action model.
Finally, in $\cI[\IS]$, we get an edge between $(v_B,E_B)$ and $(v_W,E_W)$ whenever there is an action $t = (\alpha, \sigma)$ with $\sigma := \{v_B, v_W\}$ and $t \in E_B$ and $t \in E_W$. Then by definition of our bijection between the vertices, this edge is sent to the pair $\{ (B, \view_B(\alpha, \sigma)), (W, \view_W(\alpha, \sigma)) \}$, which is also an edge.
\qed
\end{proof}

\subsection{Solvability of the equality negation task}
\label{sec:app:negation}

The equality negation task has been described in the DEL framework in Section~\ref{sec:negationTask}. For completeness, and comparison of the two approaches,
we study it here using the usual combinatorial topology approach.
Also, we present a direct proof of the impossibility of solving it in the layered message-passing model.

\subsubsection{Definition of equality negation}
Recall that in the equality negation task for two processes, $\Agents = \{\B,\W\}$,
each process starts with an input value in the set $\{0, 1, 2\}$, and has
to irrevocably decide on a value~$0$ or~$1$, such that the decisions of the two
processes are the same if and only if their input values are different.

The equality negation task is fomralized as a triple $(\cI, \cO, \Delta)$.
The input complex,  $\I$, has vertices
 $\vertices{\I} = \Agents \times \{0, 1, 2\}$, and  edges of the form $\{(B,i), (W,j)\}$ for all $i, j\in \{0, 1, 2\}$.
Thus, each vertex is a pair, whose first component is the name of a process, and
whose second component is the input value of the process.
Similarly, the output complex, $\O$, has vertices
 $\vertices{\O} = \Agents \times \{0, 1\}$, and edges of the form $\{(B,i), (W,j)\}$ for all $i, j\in \{0, 1\}$.
%
%
The carrier map $\Delta$ is as follows
\begin{equation*} \label{eq1}
\begin{split}
\Delta (\set{ (P,i)}) & =\set{(P,0),(P,1)}, P\in\set{B,W}  \\
\Delta (\set{ (B,i),(W,i)})  & =  \set{ \set{(B,0),(W,1)},\set{(B,1),(W,0)} } \cup \vertices{\O}\\
\Delta (\set{ (B,i),(W,j)})  & =  \set{ \set{(B,0),(W,0)},\set{(B,1),(W,1)} }\cup \vertices{\O}, i\neq j
\end{split}
\end{equation*}
Notice that $\Delta$ satisfies the monotonicity property and is  name-preserving.
\begin{center}
\begin{tikzpicture}[auto, scale=1.5]

  \node[label={Input complex}] (I) at (0,0) {
  \begin{tikzpicture}[auto, scale=1.5, whitevertex/.style={draw=black,thick,circle,fill=white,inner sep=2pt,minimum size=1em}, blackvertex/.style={draw=black,thick,circle,fill=black,inner sep=2pt,minimum size=1em,text=white}]
  \node[whitevertex] (w0) at (0,0) {$0$};
  \node[whitevertex] (w1) at (1,1) {$1$};
  \node[whitevertex] (w2) at (2,0) {$2$};
  \node[blackvertex] (b0) at (0,1) {$0$};
  \node[blackvertex] (b1) at (1,0) {$1$};
  \node[blackvertex] (b2) at (2,1) {$2$};
  \path (w0) edge (b0) edge (b1) edge (b2)
	    (w1) edge (b0) edge (b1) edge (b2)
        (w2) edge (b0) edge (b1) edge (b2);
  \end{tikzpicture}
  };
  
  \node[label={Output complex}] (O) at (5,0) {
  \begin{tikzpicture}[auto, scale=1.5, whitevertex/.style={draw=black,thick,circle,fill=white,inner sep=2pt,minimum size=1em}, blackvertex/.style={draw=black,thick,circle,fill=black,inner sep=2pt,minimum size=1em,text=white}]
  \node[whitevertex] (w0) at (0,0) {$0$};
  \node[whitevertex] (w1) at (1,1) {$1$};
  \node[blackvertex] (b0) at (0,1) {$0$};
  \node[blackvertex] (b1) at (1,0) {$1$};
  \path (w0) edge (b0) edge (b1)
	    (w1) edge (b0) edge (b1);
  \end{tikzpicture}
  };
  
  \draw[-latex, line width=2mm, shorten >=8pt, shorten <=8pt, blue!30, text=black] (I) -- node {Carrier map $\Delta$} (O);
\end{tikzpicture}
\end{center}

\begin{remark}
Notice that we do not have an analogue of Lemma~\ref{lem:model_isomorphic} here: the output complex is distinct from the product update model $\cI[\cT]$ of Section~\ref{sec:DEL_analysis}.
\end{remark}

\subsubsection{Equality negation impossibility}
\paragraph{Direct impossibility}

Here we show that there is no solution to the equality negation task 
in the layered message-passing model, using the usual combinatorial topology approach.

Assume for contradiction, that there is a solution, with some number of layers, $N$.
Then, by Theorem~\ref{th:mainSubdiv}, 
the protocol graph $\cP$
for the $N$-layer message passing model is a subdivision of the input graph, where each edge is divided into~$3^N$ edges.
The  execution carrier map $\Xi$ from $\cI$ to $\cP$,
 carries each input edge to its subdivision in the protocol graph, and
 each input vertex $(P,v)$ to the solo execution in which $P$
finishes the protocol without hearing from the other process.
Also,  Definition~\ref{def:solving} states that
if $\cP$ solves equality negation, then  
there exists a simplicial {decision map} $\delta$ from $\cP$ to $\cO$
such that $\delta \circ \Xi$ is carried by $\Delta$.

Consider the subgraph $\cG$ of $\I$ induced by all edges that have distinct input values.
Notice that $\cG$ is connected. Thus, $\Xi(\cG)$, is connected, by Theorem~\ref{th:mainSubdiv}.
Recall that Lemma~\ref{lem:connectSM} states that the
 image of a connected graph  under a simplicial map is
connected.
Thus, $\delta(\Xi(\cG))$ is a connected subgraph of $\O$.

The specification of the equality negation task states that the decisions should
be equal, for all of $\cG$. But the subgraph of $\O$ of edges with the same
decision is disconnected. Therefore,  $\delta(\Xi(\cG))$ must be equal to one of the edges
in this subgraph, without loss of generality, 
$\delta(\Xi(\cG))$ is equal to the subgraph of $\O$ with the single edge  $\set{(B,0),(W,0)}$.

It follows that 
\begin{equation} \label{eq2}
\begin{split}
\delta(\Xi(B,0)) & = (B,0) \\
\delta(\Xi(W,0)) & = (W,0) .
\end{split}
\end{equation}
Now consider the input edge $ \set{(B,0),(W,0)} $.
Notice that  $\Xi(  \set{(B,0),(W,0)}  )$ is connected, by Theorem~\ref{th:mainSubdiv}.
Thus, $\delta (\Xi(  \set{(B,0),(W,0)}  ))$ is equal to one of the connected subgraphs
of $\O$ with distinct output values, without loss of generality,
the one with edge $ \set{(B,0),(W,1)}$. Which implies that $\delta(\Xi(W,0))  = (W,1)$.
But this contradicts \ref{eq2} above, which states that $\delta(\Xi(W,0))  = (W,0)$.
 

\paragraph{Impossibility by reduction to consensus}
We  describe the impossibility of solving equality negation by reduction
to consensus, presented in~\cite{DBLP:journals/siamcomp/LoH00}.
The setting there is a shared memory system where processes communicate by reading
and writing shared registers, which is equivalent to our
layered message-passing model. 

The proof of~\cite{DBLP:journals/siamcomp/LoH00} is as follows.
They call the two processes $P_0,P_1$. 
They assume for  contradiction,
that  there is an equality negation algorithm, $Aen$, in the layered message-passing 
model. 
Then,  they consider a \emph{solo} execution of each of the two processes, $P_0,P_1$,
of $Aen$, where $P_k$ produces an output value without ever having heard
of the input value of the other process.
Noticing that a solo execution is deterministic, when $P_k$ executes alone
with input value $v$, 
and the output value of $P_k$ is uniquely defined, denoted $\Delta_k^v$.
Then they observe that there are inputs $\alpha_0,\alpha_1 \in\{ 0,1,2\}$ for $P_0,P_1$ resp.
such that either 
\begin{description}
\item[(i)] $\alpha_0 = \alpha_1$ and
$\Delta_0^v = \Delta_1^v$ or 
\item[(ii)]  $\alpha_0\neq \alpha_1$ and
$\Delta_0^v\neq \Delta_1^v$.
\end{description}

Then, they describe the following algorithm that solves consensus, based on $Aen$,
for two processes $Q_0,Q_1$, each one starting with its own input $inp_k$
and deciding a consensus value $out_k$.
Processes execute $Aen$, each one, $Q_k$, invoking it with input $\alpha_k$,
and they each get back an equality negation value, $w_k$.
Recall that the layered message-passing model is full-information,
so when processes execute $Aen$, they also pass to each other their inputs, $inp_k$.
Finally, process $Q_k$ decides a consensus value $out_k$ according to the rule:
\begin{enumerate}
\item \label{first}
if $w_k= \Delta_k^{\alpha_k}$ then decide its own input, $out_k:=inp_k$,
 \item \label{nosolo}
 else decide the input value of the other process, $out_k:=inp_{(1-k)}$.
 \end{enumerate}
 The proof is as follows.
 For validity, we need to check only case \ref{nosolo}., and to notice
 that the execution of $Aen$ could not have been a solo execution, and hence indeed~$Q_k$
 has received the input value of the other process.
 To prove agreement, consider the values $w_0,w_1$ in an execution where both processes
 decide. Consider the execution simulated of $Aen$ by $Q_0,Q_1$: 
 by the negation property, $w_0=w_1$ iff $\alpha_0\neq \alpha_1$. 
 But by definition of $\alpha_0,\alpha_1$,  $w_0=w_1$ iff  $\Delta_0^{\alpha_0}\neq \Delta_1^{\alpha_1}$.
  Thus, if one process $Q_k$ decides according to case~\ref{first}, the other process has to decide according to case~\ref{nosolo}, and reciprocally. So, they both decide the same output.
  
 
 Thus, there is a consensus algorithm in the layered message-passing model, 
contradicting the classic   consensus impossibility of~\cite{Chor:1987:PCU:41840.41848,LA87}.

\end{document}